\documentclass[conference]{IEEEtran}

\usepackage{bm}
\usepackage{cite}
\ifCLASSINFOpdf
\else
\fi
\usepackage[cmex10]{amsmath}
\usepackage{array}
\usepackage{fixltx2e}
\usepackage{url}
\usepackage{amsthm}

\theoremstyle{plain}
\newtheorem{theorem}{Theorem}

\newtheorem{lemma}[theorem]{Lemma}
\newtheorem{corollary}[theorem]{Corollary}

\theoremstyle{definition}

\newtheorem{fact}[theorem]{Fact}
\newtheorem{example}[theorem]{Example}

\newcommand{\longv}[1]{}

\newtheorem{thm}[theorem]{Theorem}
\newtheorem{prop}[theorem]{Proposition}
\newtheorem{cor}[theorem]{Corollary}

\newtheorem{deff}[theorem]{Definition}

\newtheorem{rem}[theorem]{Remark}

\newenvironment{varitemize}
{
\begin{list}{\labelitemi}
{\setlength{\itemsep}{0pt}
 \setlength{\topsep}{0pt}
 \setlength{\parsep}{0pt}
 \setlength{\partopsep}{0pt}
 \setlength{\leftmargin}{15pt}
 \setlength{\rightmargin}{0pt}
 \setlength{\itemindent}{0pt}
 \setlength{\labelsep}{5pt}
 \setlength{\labelwidth}{10pt}
}}
{
 \end{list} 
}

\newcounter{number}

\newenvironment{varenumerate}
{\begin{list}{\arabic{number}.}
  {
   \usecounter{number}
   \setlength{\labelwidth}{4.0mm}
   \setlength{\labelsep}{2.0mm}
   \setlength{\itemindent}{0.0mm}
   \setlength{\itemsep}{0.0mm}
   \setlength{\topsep}{0.0mm}
   \setlength{\parskip}{0.0mm}
   \setlength{\parsep}{0.0mm}
   \setlength{\partopsep}{0.0mm}
  }
}
{\end{list}}

\newcommand{\tens}{\otimes}
\newcommand{\lin}{\multimap}

\newcommand{\PCF}{\ensuremath{\mathsf{PCF}}}

\newcommand{\MLL}{\ensuremath{\textsf{MLL}}}

\newcommand{\typeof}{0} %

\usepackage{ifthen}
\newcommand{\condinc}[2]{\ifthenelse{\equal{\typeof}{0}}{#1}{#2}}

\usepackage{graphicx}

\usepackage{tikz}
\usepackage{mathrsfs}
\usepackage{amssymb}

\usepackage{times}
\usepackage{cmll}
\usepackage{stmaryrd}
\usepackage{proof}

\renewcommand{\b}{{}^{\bot}}

\newcommand{\SIAM}{\textsf{SIAM}}

\newcommand{\MELL}{\textsf{MELL}}
\newcommand{\SMLL}{\textsf{SMLL}}
\newcommand{\SMELL}{\textsf{MELLS}}
\newcommand{\SMELLY}{\textsf{SMEYLL}}

\newcommand{\one}{\mathsf{1}}
\newcommand{\botlk}{\mathsf{bot}}
\newcommand{\onelk}{\mathsf{one}}

\newcommand{\cutlk}{\mathsf{cut}}
\newcommand{\axlk}{\mathsf{ax}}

\newcommand{\slk}{\mathsf{sync}}
\newcommand{\dlk}{\mathsf{?d}}

\newcommand{\up}{\uparrow}
\newcommand{\down}{\downarrow}
\newcommand{\stable}{\leftrightarrow}
\newcommand{\dr}{\mathtt{dir}}

\newcommand{\atomone}{\alpha}

\newcommand{\formone}{A}
\newcommand{\formtwo}{B}

\newcommand{\bnf}{::=}
\newcommand{\midd}{\; \; \mbox{\Large{$\mid$}}\;\;}

\newcommand{\posfone}{P}
\newcommand{\negfone}{N}

\newcommand{\netone}{R}
\newcommand{\nettwo}{S}

\newcommand{\red}{\rightsquigarrow}
\newcommand{\rednet}{\red}

\newcommand{\ie}{\emph{i.e.}}
\newcommand{\eg}{\emph{e.g.}}

\newcommand{\ONES}{\mathtt{ONES}}

\newcommand{\DEREL}{\mathtt{DER}}
\newcommand{\DER}{\mathtt{DER}}
\newcommand{\START}{\mathtt{START}}
\newcommand{\PDOORS}{\mathtt{STABLE}}
\newcommand{\id}{\mathtt{CopyID}_{\st}}

\newcommand{\POSALL}{\mathtt{POS}} 
\newcommand{\POSI}{\mathtt{INIT}}
\newcommand{\POSF}{\mathtt{FIN}}

\newcommand{\stI}{\mathbf{I}}

\newcommand{\sem}[1]{[\![#1]\!]}

\usepackage{xcolor}

\newcommand{\todo}[1]{{ \color{red}{#1}}}

\newcommand{\blue}[1]{{\color{blue}{#1}}}

\newcommand{\BB}{b}

\newcommand{\topnodes}{\mathcal T}

\newcommand{\bbox}{$\bot$-box}
\newcommand{\bboxes}{$\bot$-boxes}

\newcommand{\st}{\mathbf{T}}
\newcommand{\scod}{\mathtt{Current_{\st}}}
\newcommand{\sdom}{\mathtt{Dom_{\st}}}
\newcommand{\pos}{\mathtt{pos_{\st}}}

\newcommand{\sttwo}{\mathbf{U}}

\newcommand{\postwo}{\mathtt{pos_{\sttwo}}}

\newcommand{\ones}{\mathtt{ONES}}

\newcommand{\pp}{\mathbf{p}}
\newcommand{\qq}{\mathbf{q}}
\renewcommand{\ss}{\mathbf{s}}

\renewcommand{\b}{{}^{\bot}}

\newcommand{\encode}[2]{\lceil #1, #2 \rceil}

\newcommand{\stk}{\mathord{s}}
\newcommand{\bstk}{\mathord{t}}
\newcommand{\edg}{\mathord{e}}

\newcommand{\emp}{\epsilon}

\newcommand{\M}{\mathcal{M}}
\newcommand{\machine}[1]{\M_{#1}}
\newcommand{\redsiam}{\rightarrow}
\newcommand{\stopsiam}{\nrightarrow}
\newcommand{\sts}{\mathcal{S}}
\newcommand{\trsf}{\mathrm{trsf}}
\newcommand{\orig}{\mathrm{orig}}
\newcommand{\partto}{\rightharpoonup}

\renewcommand{\sem}[1]{\llbracket{#1}\rrbracket}

\newcommand{\tensor}{\otimes}
\newcommand{\bang}{\mathop{!}\nolimits} 

\newcommand{\separ}{\,|\,}

\newcommand{\typeone}{A}
\newcommand{\typetwo}{B}
\newcommand{\typethree}{C}

\newcommand{\PCFlproj}{\pi_l}
\newcommand{\PCFrproj}{\pi_r}
\newcommand{\PCFpair}[1]{{\langle{#1}\rangle}}
\newcommand{\PCFn}[1]{\overline{#1}}
\newcommand{\PCFzero}{\PCFn{0}}
\newcommand{\PCFsucc}{\mathtt{s}}
\newcommand{\PCFpred}{\mathtt{p}}
\newcommand{\PCFifzero}[3]{{{\tt if}\,{#1}\,{\tt then}\,{#2}\,{\tt
      else}\,{#3}}}
\newcommand{\PCFletrec}[4]{{{\tt letrec}\,{#1}\,{#2}={#3}\,{\tt in}\,{#4}}}

\newcommand{\PCFnat}{\mathbb{N}}
\newcommand{\PCFarrow}{\to}
\newcommand{\PCFprod}{\times}

\newcommand{\PCFentail}{\vdash}
\newcommand{\PCFcbn}{\to_{\it cbn}}
\newcommand{\PCFcbv}{\to_{\it cbv}}

\usepackage{gensymb}
\begin{document}
%
\title{Parallelism and Synchronization\\ in an Infinitary Context}

\author{
\IEEEauthorblockN{Ugo Dal Lago}
\IEEEauthorblockA{Univ. di Bologna \& INRIA
}
\and
\IEEEauthorblockN{Claudia Faggian}
\IEEEauthorblockA{CNRS \& Univ. Paris Diderot
}
\and
\IEEEauthorblockN{Beno\^it Valiron}
\IEEEauthorblockA{CentraleSup\'{e}lec \& LRI, Univ. Paris Sud
}
\and
\IEEEauthorblockN{Akira Yoshimizu}
\IEEEauthorblockA{Univ. of Tokyo
}
}

\maketitle

\begin{abstract}
  We study multitoken interaction machines in the context of a very
  expressive linear logical system with exponentials, fixpoints and
  synchronization. The advantage of such machines is to provide models
  in the style of the Geometry of Interaction, \ie, an interactive
  semantics which is close to low-level implementation. On the one
  hand, we prove that despite the inherent complexity of the
  framework, interaction is guaranteed to be deadlock-free. On the
  other hand, the resulting logical system is powerful enough to embed
  \PCF{} and to adequately model its behaviour, both when call-by-name
  and when call-by-value evaluation are considered. This is not the
  case for single-token stateless interactive machines.
\end{abstract}

\IEEEpeerreviewmaketitle

\section{Introduction}
What is the inherent parallelism of higher-order functional programs?
Is it possible to turn $\lambda$-terms into low-level programs, at the
same time exploiting this parallelism? Despite great advances in very
close domains, these questions have not received a definite answer, 
yet. The main difficulties one faces when dealing with parallelism
and functional programs are due to the higher-order nature of those
programs, which turns them into objects having a non-trivial
interactive behaviour.

The most promising approaches to the problems above are based on Game
Semantics~\cite{HylandO00,AbramskyJM00} and the Geometry of
Interaction~\cite{Girard89} (GoI), themselves tools which were
introduced with purely semantic motivations, but which have later been
shown to have links to low-level formalisms such as asynchronous
circuits~\cite{GhicaSS11}. This is especially obvious when Geometry of
Interaction is presented in its most operational form, namely as a
token machine~\cite{DanosRegnier}.

Most operational accounts on the Geometry of Interaction are in
\emph{particle-style}, i.e., a \emph{single} token travels around the
net; this is largely due to the fact that parallel computation without
any form of synchronization nor any data sharing is not particularly
useful, so having multiple tokens would not add anything to the
system. While some form of synchronization was implicit in earlier
presentations of GoI, the latter has been given a proper status only
recently, with the introduction of \SMLL~\cite{lics2014}, where
\emph{multiple} tokens circulate simultaneously, and also
\emph{synchronize} at a new kind of node, called a \emph{sync
  node}. All this has been realized in a minimalistic logic, namely
multiplicative linear logic, a logical system which lacks any copying
(or erasing) capability and, thus, is not an adequate model of realistic
programming languages (except purely linear ones, whose role is
relevant in quantum computation~\cite{SelingerValiron05}).

Multitoken GoI machines are relatively straightforward to define in a
linear setting: all \emph{potential} sources of parallelism give rise
to \emph{actual} parallelism, since erasing and copying are simply
forbidden. As a consequence, managing parallelism, and in particular
the spawning of new tokens, is easy: the mere syntactical occurrence
of a source of parallelism triggers the creation of a new
token. Concretely, these sources of parallelism are \emph{unit nodes}
(when thought logically), or \emph{constants} (when read through the lenses of
functional programming). The reader will find an example in
Section~\ref{sect:multexpo}, Fig.~\ref{fig:potential}.

But can all this scale to more expressive proof theories and
programming formalisms? If programs or proofs are allowed to copy or
erase portions of themselves, the correspondence between potential and
actual parallelism vanishes: any occurrence of a unit node can
possibly be erased, thus giving rise to \emph{no} token, or copied,
thus creating \emph{more than one} token. The underlying interactive
machinery, then, necessarily becomes more complex. But \emph{how}? The
solution we propose here relies on linear logic itself: it is the way
copying and erasing are handled by the exponential connectives of
linear logic which gives us a way out.  We find the resulting theory
simple and elegant.

In this paper we generalize the ideas behind \SMLL{} in giving a
proper status to synchronization and parallelism in GoI. We show that
multiple tokens and synchronization can work well together in a
\emph{very expressive } logical system, namely multiplicative linear
logic with \emph{exponentials}, \emph{fixpoints}, and
\emph{units}. The resulting system, called \SMELLY, is then
general enough to simulate universal models of functional programming:
we prove that \PCF{} can be embedded into \SMELLY, both when
call-by-name and call-by-value evaluation are considered.  The latter
is not the case for single-token machines, as we illustrate in
Section~\ref{sect:multexpo}.

This is a   version extended with proofs and more details  of an eponymous paper \cite{lics2015} which appeared in the
proceedings of the Thirteenth Annual Symposium on Logic in Computer Science .

\subsection*{Contributions}
This paper's main contributions can be summarized as follows:
\begin{varitemize}
\item 
  \emph{An Expressive Logical System.} We introduce \SMELLY{} nets,
  whose expressiveness is increased over $\MELL$ nets by several
  constructs: we have \emph{fixpoints} (captured by the $Y$-box), an
  operator for \emph{synchronization} (the sync node), and a
  \emph{primitive conditional} (captured by the $\bot$-box).  The
  presence of fixpoints forces us to consider a restricted notion of
  reduction, namely closed \emph{surface reduction} (\ie, reduction
  never takes place inside a box).  Cuts can \emph{not} be eliminated
  (in general) from \SMELLY{} proofs, as one expects in a system with
  fixpoints.  Reduction, however, is proved to be
  \emph{deadlock-free}, \ie, normal forms cannot contain surface cuts.
\item 
  \emph{A Multitoken Interactive Machine.} \SMELLY{} nets are seen as
  interactive objects through their synchronous interactive abstract
  machine (\SIAM{} in the following). Multiple tokens circulate around
  the net simultaneously, and synchronize at sync nodes.  We prove
  that the \SIAM{} is an \emph{adequate computational model}, in the
  sense that it precisely reflects normalization through machine
  execution.  The other central result about the \SIAM{} is
  \emph{deadlock-freeness}, \ie, if the machine terminates it does so
  in a final state. In other words, the execution does not get stuck,
  which in principle could happen as we have several tokens running in
  parallel and to which we apply guarded operators (\eg,
  synchronization). Our proof comes from the interplay of nets and
  machines: we transfer \emph{termination} from machines to 
  nets, and then transfer \emph{back deadlock-freeness} from nets to
  machines.
\item 
  \emph{A Fresh Look at CBV and CBN.} 
  A slight variation on {\SMELLY} nets, and the corresponding notion
  of interactive machine, is shown to be an adequate model of
  reduction for Plotkin's \PCF~\cite{Plotkin}. This works both for call-by-name and
  call-by-value evaluation and, noticeably, the \emph{same}
  interactive machine is shown to work in \emph{both} cases: what
  drives the adoption of each of the two mechanisms is, simply, the
  translation of terms into proofs. What is surprising here is that
  CBV can be handled by a stateless interactive machine, even without
  the need to go through a CPS translation. This is essentially due to
  the presence of multiple tokens.
\item 
  \emph{New Proof Techniques.} {Deadlock-freeness} is a key issue when
  working with multitoken machines.  A direct scheme to prove it (the
  one used in \cite{lics2014}) would be: (i) prove cut elimination for
  the nets, (ii) prove soundness for the machine, and (iii) deduce
  deadlock-freeness from (i) and (ii).  However, in a setting with
  fixpoints, cut elimination is not available because termination
  simply does not hold\footnote{Even without fixpoints, there is to
    the authors' knowledge no direct combinatorial proof of
    termination for surface reduction.}.  Instead, we develop a new
  technique, which heavily exploit the interplay between net rewriting
  and the multitoken machine.  Namely, we \emph{transfer} termination
  of the machine (including termination as a deadlock) into
  termination of the nets. This combinatorial technique is novel and
  uses multiple tokens in an essential way. It appears to be of
  technical interest in its own.
\end{varitemize}
\subsection*{Related Work}
Almost thirty years after its introduction, the literature on GoI is
vast. Without any aim of being exhaustive, we only mention the works
which are closest in spirit to what we are doing here.

The fact that  GoI can be turned into an implementation scheme for
purely functional (but expressive) $\lambda$-calculi, has been
observed since the beginning of the nineties~\cite{DanosRegnier,Mackie95}. Among
the different ways GoI can be formulated, both (directed) virtual
reduction and bideterministic automata have been shown to be amenable
to such a treatment. In the first case, parallel
implementations~\cite{PediciniQ07,Pinto01} have also been introduced. We
claim that the kind of parallel execution we obtain in this work is
different, being based on the underlying automaton and not on virtual
reduction.

The fact that GoI can simulate call-by-name evaluation is well-known,
and indeed most of earlier results relied on this notion of reduction.
As in games~\cite{AbramskyM97}, call-by-value requires a more
sophisticated machinery to be handled by GoI. This machinery, almost
invariably, relies on effects~\cite{HoshinoMH14,Schopp14}, even when
the underlying language is purely functional. This paper suggests an
alternative route, which consists in making the underlying machine
parallel, nodes staying stateless.

Another line of work is definitely worth mentioning here, namely Ghica
and coauthors' Geometry of Synthesis~\cite{Ghica07,GhicaS10}, in which
GoI suggests a way to compile programs into circuits. The obtained circuit,
however, is bound to be sequential, since the interaction
machinery on which everything is based is particle-style.

On the side of nets, Y-boxes allow us to handle \emph{recursion}. A
similar box was originally introduced by
Montelatici~\cite{Montelatici03}, even though in a polarized
setting. Our Y-box differs from it both in the typing and in the
dynamics; these differences are what make it possible to build a GoI
model.

\section{On Multiple Tokens and the Exponentials}\label{sect:multexpo}
In this section, we will explain through a series of examples
\emph{how} one can build a multitoken machine for a non-linear typed
$\lambda$-calculus, and \emph{why} this is not trivial.

Let us first consider a term computing a simple arithmetical
expression, namely $M=(\lambda x.\lambda y.x+y)(4-2)(1+2)$. This term
evaluates to $5$ and is purely linear, i.e. the variables $x$ and $y$
appear exactly once in the body of the abstraction.  How could one
evaluate this term trying to exploit the inherent parallelism in it?
Since we \emph{a priori} know that the term is linear, we know that
the subexpressions $S=(4-2)$ and $T=(1+2)$ are indeed needed to
compute the result, and thus can be evaluated in parallel. The
subexpression $x+y$ could be treated itself this way, but its
arguments are missing, and should be waited for. What we have just
described is precisely the way the multitoken machine for \SMLL{}
works~\cite{lics2014}, as in Fig.~\ref{fig:potential} (left): each
constant in the underlying proof gives rise to a separate token, which
flows towards the result. Arithmetical operations act as
synchronization points.
\begin{figure}[h]
\begin{center}
\fbox{
\begin{minipage}{.47\textwidth}
\begin{center}\includegraphics[width=8cm]{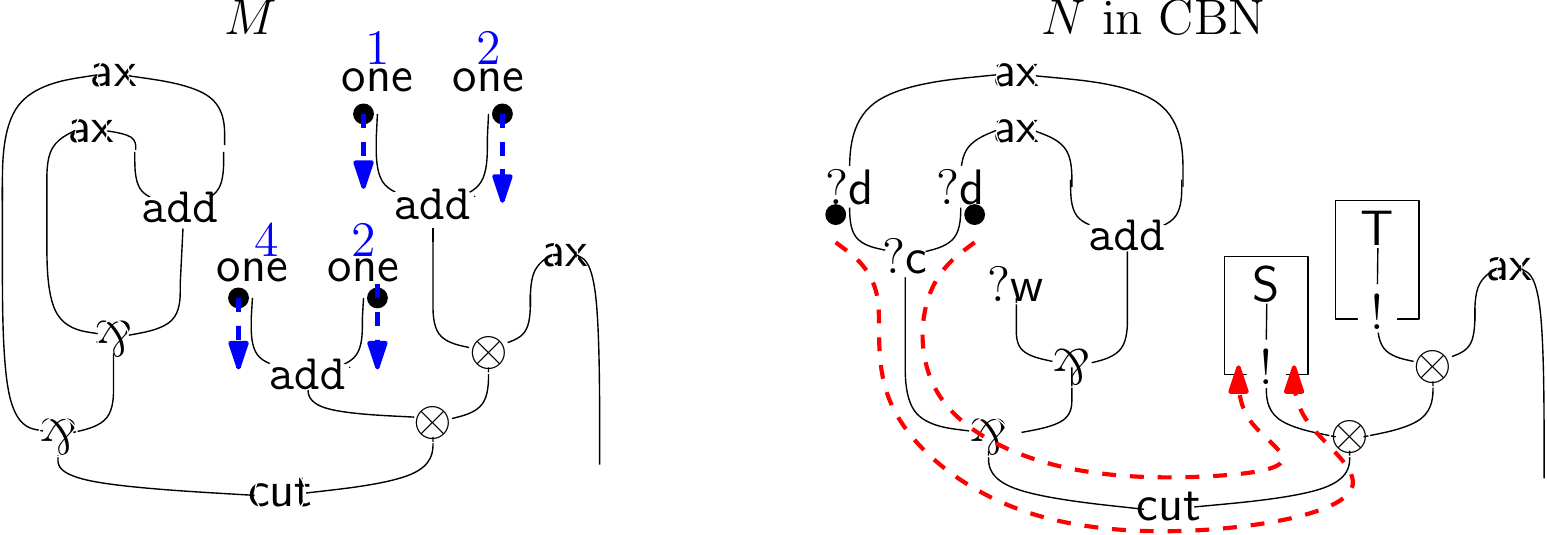}\end{center}
\end{minipage}}
\end{center}
\caption{Actual vs. Potential Parallelism.}\label{fig:potential}
\end{figure}
Now, consider a slight variation on the term $M$ above, namely
$N=(\lambda x.\lambda y.x+x)(4-2)(1+2)$. The term has a different
normal form, namely $4$, and is \emph{not} linear, for two different
reasons: on the one hand, the variable $x$ is used twice, and on the
other, the variable $y$ is not used at all. How should one proceed,
then, if one wants to evaluate the term in parallel? One possibility
consists in evaluating subexpressions \emph{only if} they are really
needed. Since the subexpression $x+x$ is of course needed (it is,
after all, the result!), one can start evaluating it. The value of the
variable $x$, as a consequence, is needed, and the subexpression it
will be substituted for, namely $4-2$, must itself be evaluated.  On
the other hand, $1+2$ should not be evaluated, simply because its value
does not contribute to the final result. This is precisely what
call-by-name evaluation actually do.  The interactive machine
which we define in this paper captures this process.  It has to be
noticed, in particular, that discovering that one of the 
subexpressions is needed, while the other is not, requires some
work. The way we handle all this is strongly related to the structure
of the exponentials in linear logic. We give the CBN translation
of $N$ in Fig.~\ref{fig:potential} (right). The two rightmost
subterms are translated into exponential boxes (where $\mathsf{S}$ is
the net for $4-2$ and $\mathsf{T}$ for $1+2$), which serve as
\emph{boundaries} for parallelism: whatever potential parallelism a
box includes, must be triggered before giving rise to an actual
parallelism. Each of the occurrences of the variable $x$ triggers a
new kind of token, which starts from the dereliction nodes ($?d$) at the 
surface  and whose purpose is precisely to look for the box the
variable will be substituted for. We call these \emph{dereliction
  tokens}.

What happens if we rather want to be consistent with
call-by-\emph{value} evaluation?  In this case, both subterms $(4-2)$
and $(1+2)$ in the term $N$ above should be evaluated. Let us however
consider a more extreme example, in which call-by-name and
call-by-value have different \emph{observable} behaviors, for example
the term $L=(\lambda x.1)\Omega$, where $\Omega=(\lambda x.xx)(\lambda
x.xx)$. The call-by-value evaluation of $L$ gives rise to divergence,
while in call-by-name $L$ evaluates to $1$.  Something extremely
interesting happens here.  We give the call-by-value translation of
$L$ in Fig.~\ref{fig:cbvomega}.
\begin{figure}[h]
\begin{center}
\fbox{
\begin{minipage}{.47\textwidth}
\begin{center}\includegraphics[width=8cm]{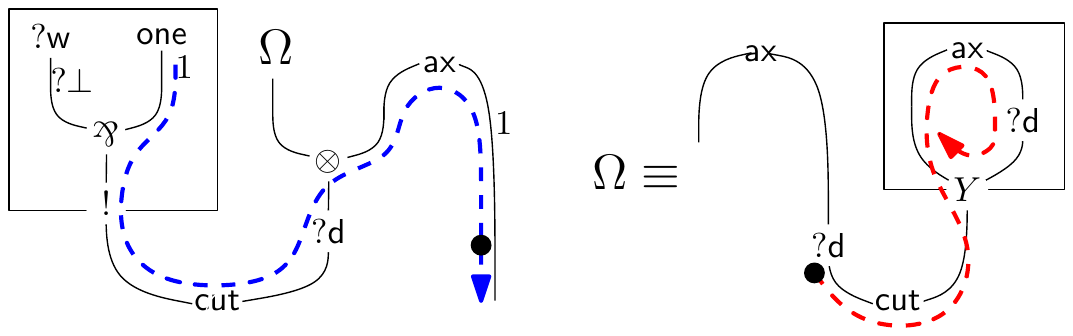}\end{center}
\end{minipage}}
\end{center}
\caption{The CBV-translation of $(\lambda x.1)\Omega$.}\label{fig:cbvomega}
\end{figure}
First of all, we observe that a standard \emph{single-token} machine
would start from the conclusion, find the node $\onelk$, and exit
again: such a machine would simply converge on the term $L$. When
running on the term $\Omega$ alone, the machine would diverge, but as
subterm of $L$, $\Omega$ is never reached, so the machine's behaviour
on $L$ is not the one which we would expect in call-by-value.  Our
multitoken machine, instead, simultaneously launches tokens from all
dereliction nodes at surface: the dereliction token coming out of
$\Omega$ (represented on the right in Fig.~\ref{fig:cbvomega}) reaches
the Y-box, and makes the machine diverge.

We end this section by stressing that the interactive machine we use
is the same, and that this machine correctly models CBN and CBV, solely
depending on the chosen translation of terms into nets.  The
call-by-name translation of $L$ puts the subterm $\Omega$ in a box
which is simply unreachable from the rest of the net (as in the case
of $\mathsf{T}$ in Fig.~\ref{fig:potential}), and our machine
converges as expected. The call-by-value translation of $L$, on the
other hand, does \emph{not} put $\Omega$ inside a box. As a
consequence, there is no barrier to the computation to which $\Omega$
gives rise---the same as if $\Omega$ would be alone---and our machine
correctly diverges. This is the key difficulty in any interactive
treatment of CBV, and we claim that the way we have solved it is
novel.

\section{Nets and a Multitoken Interactive Machine }\label{sec:SMELLY}
We start with an overview of this section, which is divided into four
parts.
\paragraph*{Nets and Their Dynamics} 
\SMELLY\ nets come with \emph{rewriting} rules, which provide an operational
semantics for them, and with a \emph{correctness} criterion, which
ultimately guarantees that nets rewriting is deadlock-free.
\paragraph*{Multitoken Machines}
On any net we define a \emph{multitoken} machine, called \SIAM, which
provides an effective computational model in the style of GoI. A
fundamental property we need to check for the machine is
\emph{deadlock-freeness}, \ie, if the machine terminates it does so in
a final state. From the beginnings of linear logic, the correctness
criterion of nets has been interpreted as deadlock-freeness in
distributed systems \cite{Asperti}; this is also the case for
\SMELL. Here, however, we work with surface reduction, and we have
fixpoints. For these reasons, a rather refined approach is needed.
\paragraph*{The Interplay Between Nets and Machines}
Nets rewriting and the \SIAM{} are tightly related. We establish the
following results. Let $R$ denote a net, $\M_R$ its machine, and
$\rednet$ the net rewriting relation. First of all, we know that (i)
if $R$ is cut-and-sync-free, the machine $\M_R$ terminates in a final
state. On the net hand, we establish that (ii) \emph{net rewriting is
  deadlock-free}: if no reduction is possible from $R$, then $R$ is
cut-and-sync-free.  On the machine side, we establish that (iii) if
$R\rednet S$ then $\M_R$ converges/deadlocks iff the same holds for
$\M_{S}$. We then use the multitoken paradigm to provide a decreasing
parameter for net rewriting, and establish that (iv) if $\M_R$
terminates, then $R$ has no infinite sequence of reductions.
Putting all this together, it follows that \emph{multitoken
  machines are deadlock-free}.
\paragraph*{Computational Semantics}
Finally, by using the machine representation, we associate a
denotational semantics to nets, which we prove to be sound with
respect to net reduction.
\subsection{Nets and Their Dynamics.}
In this section we introduce \SMELLY\ nets, which are a generalization
of \MELL\ proof nets. For a detailed account on proof nets,
we refer the reader to Laurent's notes \cite{LaurentTorino}: our
approach to correctness, as well as the way to deal with weakening, is
very close to the one described there.
\subsubsection{Formulas}
The language of \SMELLY{} \emph{formulas} is identical to the one for
\MELL. 
The language of formulas is therefore:
$$
\formone \bnf \one\midd
\bot\midd X\midd X\b\midd\formone\otimes\formone\midd\formone\parr\formone\midd!\formone\midd ?\formone,
$$ 
where $X$ ranges over a denumerable set of {propositional variables}.
The constants $1,\bot$ are the \emph{units}. \emph{Atomic formulas} are
those formulas which are either propositional variables or units.
Linear negation $(\cdot)\b$ is extended into an involution on all
formulas as usual: $A\b\b\equiv A$, $1\b\equiv \bot $,
$(\formone\otimes\formtwo)\b\equiv \formone\b\parr \formtwo\b $,
$(!\formone)\b \equiv{} ?\formone\b$.  Linear implication is a defined
connective: $\formone\lin\formtwo\equiv\formone\b\parr\formtwo$.

Atoms and connectives of linear logic are usually divided in two
classes: positive and negative. Here however, we define
\emph{positive} (denoted by $P$) and \emph{negative} (denoted by $N$)
those formulas which are built from units in the following way: $\posfone \bnf
\one \midd\posfone\otimes\posfone$, and $\negfone\bnf \bot
\midd\negfone\parr\negfone$. So in particular, there are formulas
which are neither positive nor negative, e.g.\ $\bot\parr\one$.
\subsubsection{Structures}
A \SMELLY{} \emph{structure} is a finite labeled \emph{directed} graph
built over the alphabet of nodes which is represented in
Fig.~\ref{SMELLYnets} (where the \emph{orientation} is the top-bottom
one). All edges have a source, but some edges may have no target; such
dangling edges are called the \emph{conclusions of the structure}. The
edges are labeled with \SMELLY{} formulas; the label of an edge is
called its \emph{type}.  We call those edges which are represented
below (resp. above) a node \emph{conclusions} (resp. \emph{premisses})
of the node.  We will often say that a node ``has a conclusion
(premiss) $\formone$'' as shortcut for ``has a conclusion (premiss) of
type $\formone$''.  When we need more precision, we distinguish
between an edge and its type, and we use variables such as $e,f$ for
the edges.

The nodes $! $, $Y$ and $\bot$ are called \emph{boxes}. One among
their conclusions (the leftmost ones in Fig.~\ref{SMELLYnets}, which
have type $!A$, $!A$ and $\bot$, respectively) is said to be
\emph{principal}, the other ones being \emph{auxiliary}.
$!$-boxes and Y-boxes are \emph{exponential}. An exponential box is
\emph{closed} if it has no auxiliary conclusions. To each box is associated, in an
 inductive way,  a structure which is called the
\emph{content} of the box. To the $!$-box we associate a structure
with conclusions $A,?\Gamma$. To the Y-box corresponds a structure of
conclusions $A, ?A\b, ?\Gamma$. To the {\bbox} is associated a
structure of non-empty conclusions $\Gamma$, together with a new node
$\botlk$ of conclusion $\bot$. We represent a box $b$ and its content
$S$ as in Fig.~\ref{SMELLYboxes}.  With a slight abuse of terminology,
the nodes and edges of $S$ are said to be \emph{inside} $b$.
Similarly, a crossing of any box's border is said to be a \emph{door},
and we often speak of premiss and conclusion \emph{of a} (principal or
auxiliary) \emph{door}. Note that the principal door of the Y-box
(marked by Y) has premisses $A,?A\b$ and conclusion $!A$.

A node \emph{occurs at depth 0} or \emph{at surface} in the structure
$R$ if it is a node of $R$, while it \emph{occurs at depth $n+1$} in
$R$ if it occurs at depth $n$ in a structure associated to a box of
$R$. Please observe that nets being defined inductively, the depth of
nodes is always finite.
 \begin{figure}[htbp]
   \begin{center}
     \fbox{
       \begin{minipage}{.47\textwidth}
         \begin{center}\includegraphics[width=8cm]{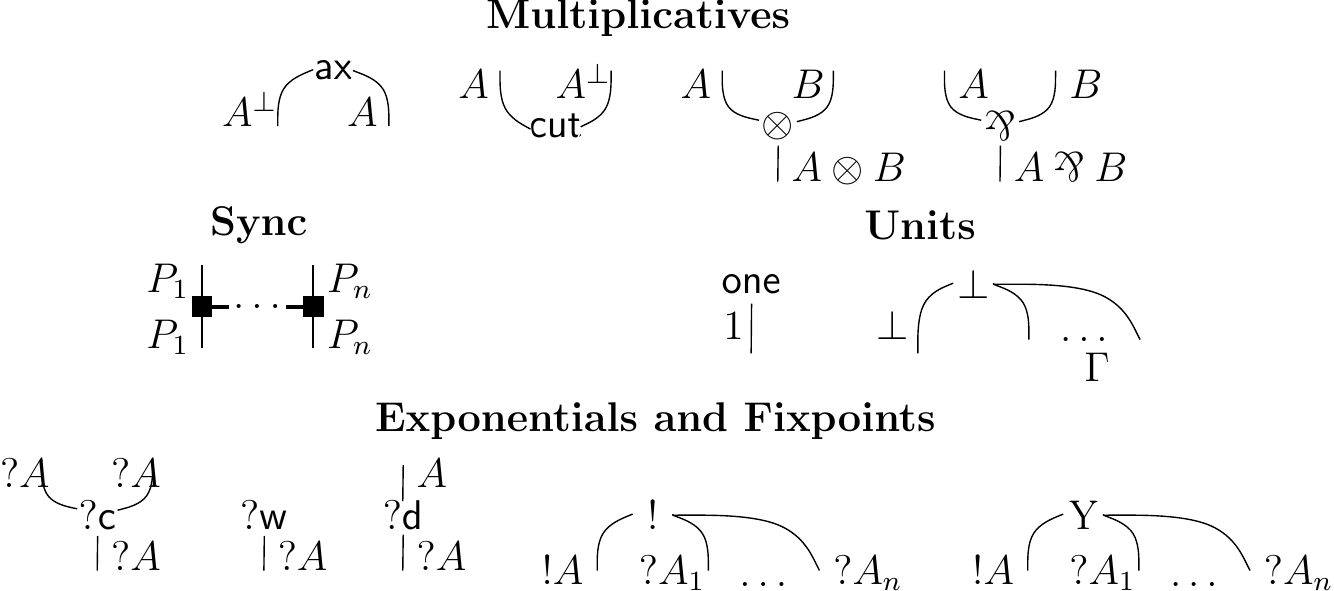}\end{center}
     \end{minipage}}
   \end{center}
   \caption{\SMELLY\ Nodes.}\label{SMELLYnets}
 \end{figure}
 \begin{figure}[htbp]
   \begin{center}
     \fbox{
       \begin{minipage}{.47\textwidth}
         \begin{center}\includegraphics[width=8cm]{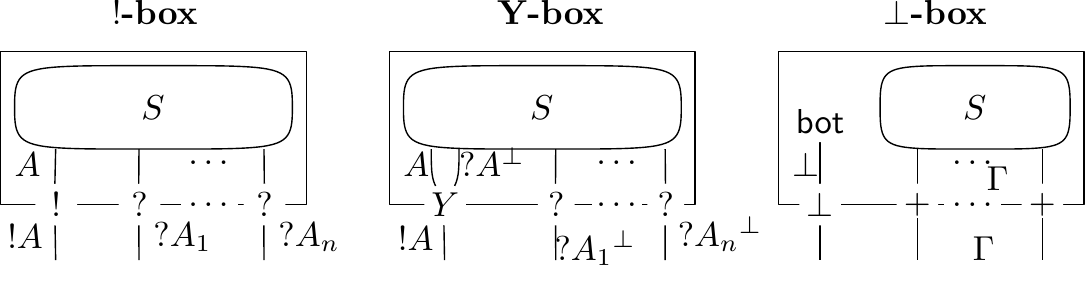}\end{center}
     \end{minipage}}
   \end{center}
   \caption{\SMELLY\ Boxes.}\label{SMELLYboxes}
 \end{figure}
The sort of each node induces constraints on the number and the labels
of its premisses and conclusions, which are shown in
Fig.~\ref{SMELLYnets}. We observe that the $\bot$-box is \emph{the
  same} as in \cite{Girard87} and corresponds to the sequent calculus
rule $\infer[]{ \vdash \bot, \Gamma }{\vdash \Gamma}$. All nodes are
standard except sync nodes and Y-boxes, which need some further
explanation:
\begin{varitemize}
\item 
  Y-boxes model \emph{recursion} (more on this when we introduce the
  reduction rules). Proof-theoretically, the Y-box corresponds to
  adding the following fix-point sequent calculus rule to \MELL: $$
  \infer[Y]{ \vdash !A,?\Gamma }{\vdash A, ?A\b,?\Gamma} $$
\item 
  Sync nodes model {\em synchronization points}. A sync node has $n$
  premisses and $n$ conclusions; for each $i$ ($1\leq i\leq n$) the
  $i$-th premiss and the \emph{corresponding} $i$-th conclusion are
  typed by the \emph{same} formula, which needs to be \emph{positive}.
\end{varitemize}

\vskip 4pt

\noindent\emph{Simple and positive structures.}
Two relevant classes of structures are  simple and positive structures. A formula is \emph{simple} it is is built out of $\{X,X\b, 1, \otimes, \parr\}$. A structure
$R$ is \emph{simple (resp. positive)} if all its conclusions are simple (resp. positive)
formulas. This \emph{does not} mean that all formulas occurring in $R$
are  simple (resp. positive). $R$ can be very complex; the
constraint only deals with $R$'s conclusions.

\subsubsection{Correctness}
A \emph{net} is a structure which fulfills a \emph{correctness
  criterion} defined by means of {switching paths} (see
\cite{LaurentTorino}).  A \emph{switching path} on the structure $R$
is an undirected path\footnote{By path, in this paper we always mean a
  \emph{simple path} (no repetition of either nodes or edges).}  such that 
(i) for each $\parr$-or-$?c$-node, the path uses at most one of the
two \emph{premisses}, and (ii) for each sync node, the path uses at
most one of the \emph{conclusions}.  The former condition is standard,
the latter condition rules out paths which bounce on sync nodes ``from
below'': a path crossing a sync node may traverse one premiss and one
conclusion, or traverse two distinct premisses. A structure is
\emph{correct} if:
\begin{varenumerate}
\item 
  none of its switching paths are cyclic, and
\item  
  the content of each of its boxes is itself correct.
\end{varenumerate}
The reader familiar with linear logic correctness criteria did
probably notice that the only condition we require is acyclicity, and
that connectedness is simply not enforced (as, e.g., in Danos and
Regnier's criterion~\cite{DanosRegnierMult}).  Actually, the only role of
connectedness consists in ruling out the so-called Mix rule from the
sequent calculus. This is not relevant in our development, so we will
ignore it. An advantage of accepting the Mix rule is that we do not
need extra conditions for dealing with weakening. A similar approach
is adopted by Laurent~\cite{LaurentTorino}.  In the following, when we
talk of \MELL{} (resp. \MLL{}), we actually always mean \MELL{} + Mix
(resp. \MLL{} + Mix).
\subsubsection{Net Reduction}
Reduction rules for nets are sketched in Fig.~\ref{SMELLYred}.
Reduction rules can be applied only at surface (\ie, when the redex
occurs at depth $0$), and not in an arbitrary context.  Moreover,
observe that reduction rules involving an \emph{exponential} box can
only be applied when the box is \emph{closed}, \ie, when it has no
auxiliary doors.  We write $\red$ for the rewriting relation induced
by these rules. Some reduction rules deserve some further
explanations:
\begin{varitemize}
\item 
  The $y$-rule unfolds a Y-box, this way modeling recursion.  The
  intuition should be clear when looking at the translation of the
  \PCF\ term $L=\PCFletrec{f}{x}{M}{N}$, which reduces to the explicit
  substitution of $f$ by $\lambda x.\PCFletrec{f}{x}{M}{M}$ in $N$,
  call it $P$. Indeed, the encoding of $L$ reduces to the encoding of
  $P$:
  \begin{center}
    \includegraphics[width=6cm]{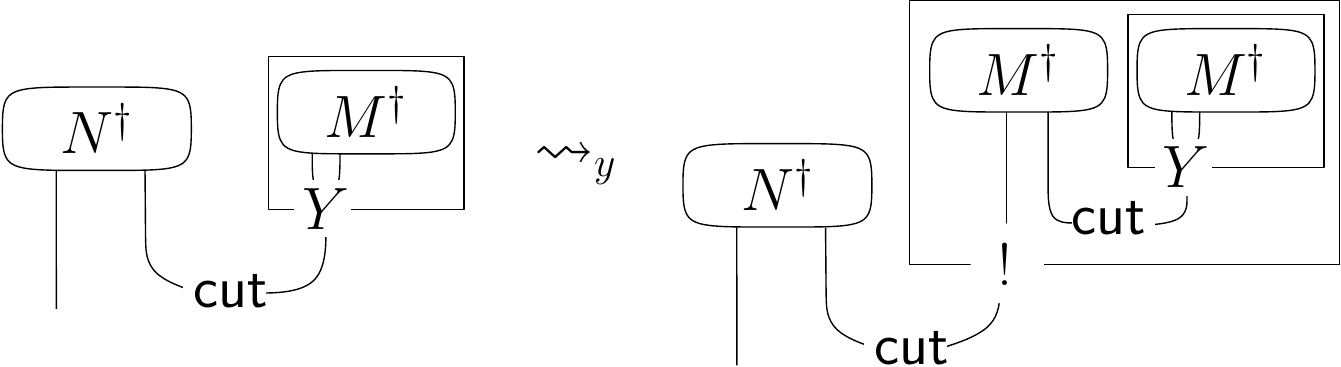}\label{fig:Yexample}
   
  \end{center}
  (where $M^\dagger$ and $N^\dagger$ stand for the encodings of $M$
  and $N$, respectively). When (and only if!) $N$
  recursively calls $f$, the corresponding $d$ node
  ``opens'' the $!$-box for the first iteration of $f$; if
  $f$ further uses a recursive call of itself, the $Y$-box again
  turns into yet another $!$-box and is opened, and so on.
\item 
  The $s.el$-rule erases a sync link whose premisses are \emph{all}
  conclusions of $\onelk$ nodes.
\item 
  The $w$-rule, corresponding to a cut with weakening, \emph{deletes}
  the redex (because the box has no auxiliary conclusions).
\item 
  The $bot.el$-rule opens a $\bot$-box.
\end{varitemize}
\begin{figure}[htbp]
  \begin{center}
    \fbox{
      \begin{minipage}{.47\textwidth}
        \begin{center}
          \includegraphics[width=8cm]{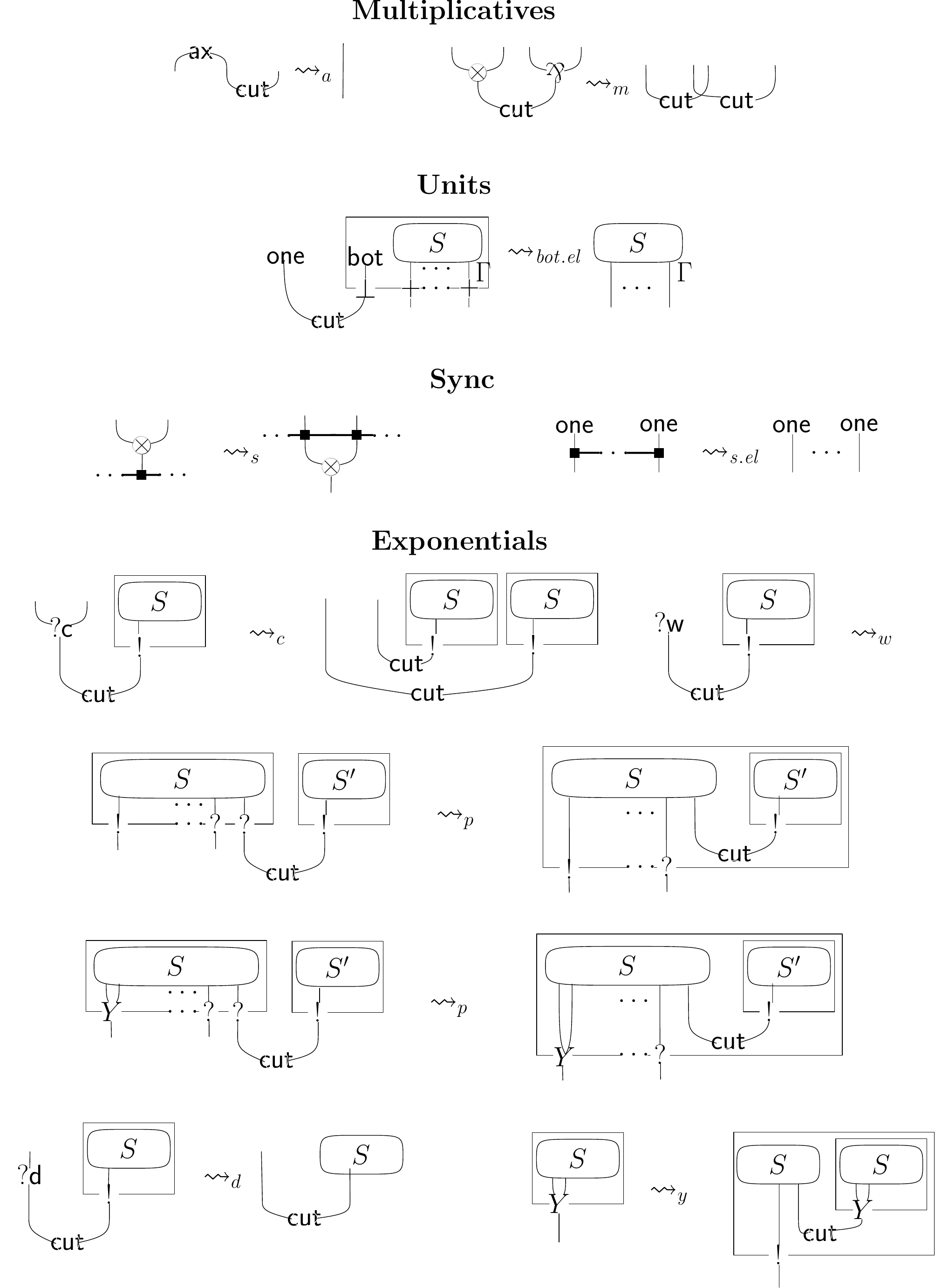}
        \end{center}
    \end{minipage}}
  \end{center}
  \caption{\SMELLY\ Net Rewriting Rules.}\label{SMELLYred}
\end{figure}
It is immediate to check that correctness is preserved by all reduction rules.
\begin{lemma}
  If $R$ is a net and $R\red S$, then $S$ is itself a net.
\end{lemma}
Since the constraints exclude most of the commutations
which are present in \MELL, rewriting enjoys a strong form of
confluence:
\begin{prop}[Confluence and Uniqueness of Normal Forms]\label{net_conf}
  The rewriting relation $\red$ has the following properties:
  \begin{varenumerate}    
  \item 
    it is confluent and normal forms are unique;
  \item 
    any net weakly normalizes iff it strongly normalizes.
\end{varenumerate}
\end{prop}
\begin{proof}
The only critical pairs are the trivial ones of \MLL, leading to the
same net. Therefore, reduction enjoys a  diamond property (uniform
confluence): if $R \red S$ and $R\red T$, then either $S=T$ or
there exists $U$ such that $S\red U$ and $T\red U$.
(1) and (2) are direct consequences.
\end{proof}
The strict constraints on rewriting, however, render cut elimination
non-trivial: it is not obvious that a reduction step is available
whenever a cut is present. We need to prove that in presence of a cut,
there is always a valid redex (\ie, it is surface, and any exponential
box acted upon is closed). The main difficulty comes from \bboxes, as
they can hide large parts of the net, and in particular dereliction
nodes which may be necessary to fire a reduction. However, the
following establishes that as long as there are cuts or syncs, it is
always possible to perform a valid reduction.
\begin{thm}[Deadlock-Freeness for Nets]\label{main_lem} 
  Let $R$ be a simple \SMELLY{} net. If $R$ contains cuts or sync nodes,
  then  a reduction applies, \ie\ there exists $S$ such that $R\red S$.
\end{thm}
The rather long proof is given in Appendix~\ref{app:SMELLY}. The key element is the
definition of an order on the boxes which occur  at depth $0$ in $R$;
the existence of such an order relies on the correctness criterion. The order
captures the dependency among boxes, \ie, exposes the order in which
cuts are eliminated.

\begin{cor}[Cut Elimination] \label{cutel} 
  Let $R$ be a simple \SMELLY\ net. If $R\red^* S$ and $S$ cannot be
  further reduced, then $S$ is a cut free \MLL{} net\footnote{Precisely, \MLL{} + Mix, as we have already pointed
      out.}, \ie, it  only containing $\axlk$,
  $\onelk$, $\otimes$, $\parr$ nodes. 
\end{cor}
\paragraph*{Discussion on simple structures}
The hypothesis which we make in Theorem \ref{main_lem}  that a structure is simple 
is an assumption which in this section we use 
\emph{to simplify auxiliary lemmas}; it will not appear in our main
result, namely Theorem~\ref{SIAM deadlock free}.
\subsection{\SIAM}\label{SIAM}
All along this section, $\netone$ indicates a \SMELLY\ structure
(with no other  hypothesis, unless otherwise stated).
\subsubsection{Preliminary Notions}
Some auxiliary definitions are needed before we can introduce our
interactive machines.  \emph{Exponential signatures} are defined by
the following grammar
$$
  \sigma \bnf {*} \midd l(\sigma) \midd r(\sigma)
  \midd \encode{\sigma}{\sigma} \midd y(\sigma,\sigma),
$$
while \emph{stacks} are defined as follows
$$
  \stk \bnf \emp \midd l.\stk \midd r.\stk \midd \sigma.\stk \midd \delta,
$$
where $\epsilon$ is the empty stack and $.$ denotes concatenation
(and, thus, $s.\epsilon=s$). Given a formula $\typeone$, a stack $\stk$
\emph{indicates an occurrence $\atomone$ of an atom} (resp. \emph{an
  occurrence $ \mu$ of a modality}) in $\typeone$ if $\stk[A]=
\atomone$ (resp. $\stk[A]= \mu$), where $\stk[A]$ is defined as
follows:
\begin{varitemize}
\item 
  $\emp [\atomone] = \atomone$,
\item 
  $\sigma.\delta[\mu B] = {\mu}$,
\item 
  $\sigma.t[\mu \typetwo] = t[\typetwo] $ whenever $t\neq\delta$,
\item 
  $l.t[\typetwo \Box \typethree]= t[\typetwo] $ and $r.t[\typetwo
  \Box \typethree]=t[\typethree]$, where $\Box$ is either $\otimes$
  or $\parr$.
\end{varitemize}
We observe that a stack can indicate a modality only if its head is
$\delta$. 
\begin{example}\label{ex:indic}
  Given the formula $A =\bang(\bot\otimes{!\one})$, the stack $*.\delta$
  indicates the first occurrence of $\bang$, $ *.r.*.\delta[A]$ gives the
  second occurrence of $!$, and $*.\delta, *.l[A]=\bot$.
\end{example}
The set of $R$'s \emph{positions} $\POSALL_R$ contains all the triples
in the form $(\edg, \stk, \bstk)$, where:
\begin{varenumerate}
\item  
  $\edg$ is an edge of $\netone$, 
\item    
  the \emph{formula stack} $\stk$ is either $\delta$ or a stack which
  indicates an occurrence of atom or modality in the type $A$ of
  $\edg$,
\item 
  the \emph{box stack} $\bstk$ is a stack of $n$ exponential
  signatures, where $n$ is the number of exponential boxes inside
  which $\edg$ appears.
\end{varenumerate} 
We use the metavariables $\ss$ and $\pp$ to indicate positions. For
each position $\pp=(\edg,\stk,\bstk)$, we define its \emph{direction}
$\dr (\pp)$ as \emph{ upwards} ($\up$) if $\stk$ indicates an
occurrence of $!$ or of negative atom, as \emph{downwards} ($\down$)
if $\stk$ indicates an occurrence of $?$ or of positive atom, as
\emph{stable} ($\stable$) if $\stk= \delta$ or if the edge $\edg$ is
the conclusion of a $\botlk$ node. A position $\pp=(\edg, \stk, \emp)$
is \emph{initial} (resp. \emph{final}) if $e$ is a conclusion of
$\netone$, and $dir(\pp)$ is $\up$ (resp. $\down$). For simplicity, on
initial (final) positions, we require all exponential signatures in
$\stk$ to be $*$. So for example, if $!(\bot\otimes{!\one})$ is a
conclusion of $R$, there is one final position ($s=*.r.*$), and three
initial positions (the three stacks given in Example~\ref{ex:indic}). The
following subsets of $\POSALL_R$ play a crucial role in the definition
of the machine:
\begin{varitemize}
\item 
  the set $\POSI_R$ of all \emph{initial positions};
\item 
  the set $\POSF_R$ of all \emph{final positions};
\item 
  the set $\ONES_{\netone}$ of positions $(\edg, \emp, \bstk)$ where
  $\edg$ is the conclusion of a $\onelk$ node;
\item 
  the set $\DEREL_{\netone}$ of positions $(\edg, *.\delta, \bstk)$
  where $\edg$ is the conclusion of a $\dlk$ node;
\item 
  the starting positions $\START_{\netone}= \POSI_R \cup
  \ONES_{\netone}\cup \DEREL_{\netone} $;
%
\item   
  the set $\PDOORS_R$ of the positions $\pp$ for which
  $\dr(\pp)=\stable$.
\end{varitemize}
The multitoken machine $\machine{\netone}$ for $\netone$
consists of a set of \emph{states} and a \emph{transition} relation
between them. These are the topics of the following two subsections.
\subsubsection{States}
A state of $\machine{\netone}$ is a snapshot description of the tokens
circulating in $\netone$. We also need to keep track of the positions
where the tokens started, so that the machine only uses each starting
position once. Formally, a \emph{state} $\st = ( \scod, \sdom)$ is a
set of positions $\scod \subseteq \POSALL_R$ together with a set of
positions $\sdom \subseteq \START_R$.  Intuitively, $\scod $ describes
the current position of the tokens, and $\sdom$ keeps track of which
starting positions have been used\footnote{In Section~\ref{tracing} we
  show that $\sdom$ is actually redundant; we have however decided to
  give it explicitly, because it makes the definition of the machine
  simpler.}. A state is \emph{initial} if $ \scod=\sdom=\POSI_R$. We
indicate the (unique) initial state of $\M_R$ by $\stI_R$. A state
$\st$ is \emph{final} if all positions in $\scod$ belong to either
$\POSF_R$ or $\PDOORS_R$. The set of all states will be denoted by
$\sts_{\netone}$. Given a state $\st$ of $\M_R$, we say that
\emph{there is a token in $\pp$} if $\pp\in \scod$.  We use
expressions such as ``a token moves'', ``crosses a node'', in the
intuitive way.
\subsubsection{Transitions}
The transition rules of $\machine{\netone}$ are given by the
transitions described in Fig.~\ref{fig:trRules} (where $\Box$ stands
for either $\otimes$ or $\parr$). The rules marked by (i)--(iii) 
make the machine concurrent, but the constraints they need to satisfy are rather
technical and for this reason we prefer to postpone the related
discussion.
\paragraph*{Transition Rules, Graphically}
The position $\pp = (\edg, \stk, \bstk)$ is represented graphically by
marking the edge $e$ with a bullet $\bullet$, and writing the stacks
$(\stk, \bstk)$.  A transition $\st \redsiam \sttwo$ is given by
depicting only the positions in which $\st$ and $\sttwo$ differ.  It
is of course intended that all positions of $\st$ which do not
explicitly appear in the picture also belong to $\sttwo$.  \condinc{}{
  \todo{ To represent a transition $\st \redsiam \sttwo$, we depict
    $\pos(\ss)$ in the left-hand-side, and $\postwo(\ss)$ on the
    right-hand-side of the arrow, for each $\ss\in\sdom$ such that
    $\postwo(\ss)\not=\pos(\ss)$.  It is of course intended that
    $\postwo(\ss)=\pos(\ss)$ for all $\ss$ whose value is not
    explicitly appearing in the picture.  }}  To save space, in
Fig.~\ref{fig:trRules} we annotate the transition arrows with a
\emph{direction}; we mean that the rule applies (only) to positions
which have that direction.  We sometimes explicitly indicate the
direction of a position by directly annotating it with
${}^{\down},{}^{\up}$ or ${}^{\stable}$. Notice that no transition is
defined for stable positions.  We observe that tokens \emph{changes
  direction} only in one of two cases: either when they move from an
edge of type $A$ to an edge of type $A\b$ (\ie, when crossing a
$\axlk$ or a $\cutlk$ node), or when they cross a $Y$-node, in the
case where the transitions are marked by (*): moving down from the
edge $A$ and then up to $?A\b$, or vice versa.  Whenever a token is on
the conclusion of a box, it can move into that box (graphically, the
token ``crosses'' the border of the box) and it is modified as if it
were crossing a node. For exponential boxes, in Fig.~\ref{fig:trRules}
we depict only the border of the box. The transitions for the
multiplicative nodes $\axlk$, $\cutlk$, $\otimes$, $\parr$ are the
standard ones.  The rules for \emph{exponential nodes} are mostly
standard. There are however two novelties: the introduction of
``dereliction tokens", \ie, tokens which start their path on the
conclusion of a $?d$ node, and the $Y$ box. We discuss both below.
\paragraph*{Some Further Comments}
Certain peculiarities of our interactive machines need to be further
discussed:
\begin{varitemize}
\item 
  \emph{Y-boxes}.  The recursive behaviour of Y-boxes is captured by
  the exponential signature in the form $y(\cdot,\cdot)$, which
  intuitively keeps track of how many times the token has entered a
  Y-box so far. Let us examine the transitions via the $Y$ door.  Each
  transition from $!A$ (conclusion of $Y$) or from $?A\b$ (premiss of
  $Y$) to the edge $A$ (premiss of $Y$) corresponds to a recursive
  call. The transition from $A$ to $?A\b$ captures the return from a
  recursive call; when all calls are unfolded, the token exits the
  box.  The auxiliary doors of a $Y$-box have the same behaviour as
  those of $!$-boxes.
\item 
  \emph{Dereliction Tokens}.  As we have explained in
  section~\ref{sect:multexpo}, this is a key feature of our machine.
  A dereliction token is generated (according to conditions (i) below)
  on the conclusion of a $\dlk$ node, as depicted in
  Fig.~\ref{fig:trRules}. Intuitively, each dereliction token
  corresponds to a copy of a box.
\item 
  \emph{Box Copies and stable tokens.} A token in a stable position is said to be
  \emph{stable}. 
   Each such token is the remains of a token which
  started its journey from $\DER$ or $\ONES$, and flowed in the graph
  ``looking for a box''. 
   This stable token that was once roaming the net
  therefore witnesses the fact that \emph{an instance} of dereliction
  or of $\onelk$ ``has found its box''.  Stable tokens play an
  essential role, as they keep track of box copies. We are going to
  formalize this immediately below.
  
  It is immediate to check that a stable token
    can only be located inside a box, more precisely on the premiss of
    its principal door.
  In  Fig.~\ref{SIAM_stable} we indicate   explicitly  all the exponential  transitions which lead to a stable  position; the other transition leading to a stable position is the one on \bbox. 
      \begin{figure}[htbp]
        \centering
        \fbox{
        \includegraphics[width=7cm]{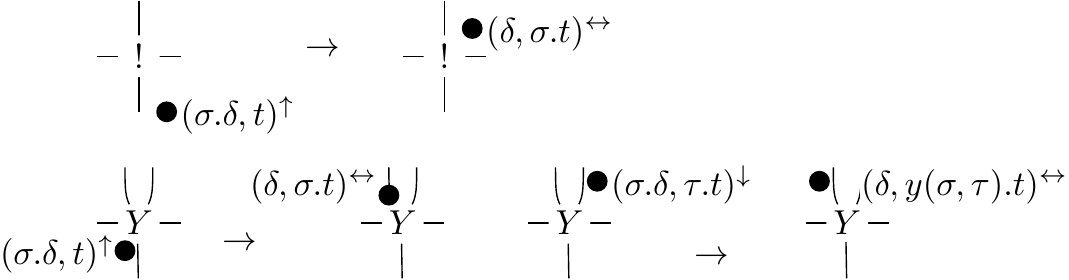}}
        \caption{Exponential transitions to a stable position }\label{SIAM_stable}
      \end{figure}
  
\end{varitemize}
\begin{figure}[htbp]
  \begin{center}
    \fbox{
    \begin{minipage}{.47\textwidth}
      \begin{center}
        \includegraphics[width=8cm]{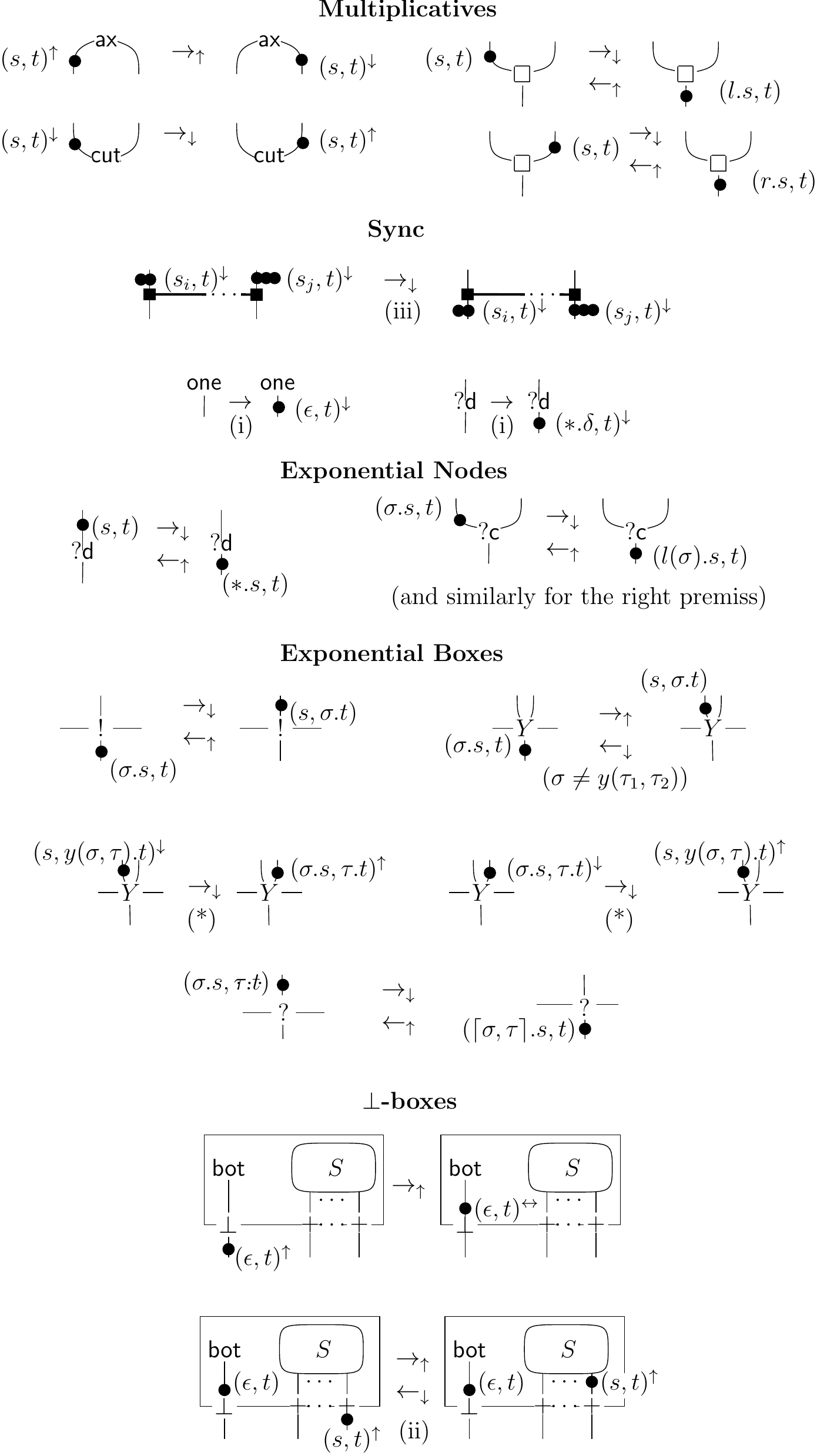}
      \end{center}
    \end{minipage}}
  \end{center}
  \caption{\SIAM{} Transition Rules.} \label{fig:trRules}
\end{figure}

\paragraph*{Multitoken Rules} 
The rules (i)--(iii) from Fig.~\ref{fig:trRules} are where the  multitoken
nature of the \SIAM{} really comes into play.  Those rules are subject to certain conditions,
which are intimately related to box copies. Given a state $\st$ of
$\M_R$, we define $\id(S)$ to be $\{\emp\}$ if $R=S$ (we are at depth 0).
Otherwise, if $S$ is the structure associated to a box node $\BB$ of $R$, we
define $\id(S)$ as the set of all $t$ such that $t$ is the box stack
of a stable token at the principal door of $\BB$.  Intuitively, as we
discussed above, the box stack of each such a token \emph{identifies a
  copy of the box} which contains $S$.
Rules marked as (i)--(iii) only apply if certain conditions are
satisfied:
\begin{varitemize}
\item[(i)] 
  The position $(\edg,\emp,\bstk)$ (resp. $(\edg,\delta,\bstk)$) does
  not already belong to $\sdom$, and $\bstk\in \id(S)$, where $S$ is
  the structure to which $e$ belongs.  If both conditions are
  satisfied, $\scod$ and $\sdom$ are extended with the position $\pp$.
  This is the only transition changing $\sdom$. Intuitively, each
  $\bstk$ corresponds to a copy of the (box containing the) $\onelk$ (resp. $?d$) node.
\item[(ii)] 
  The token moves inside the \bbox\ only if its box stack $\bstk$
  belongs to $\id(S)$, where $S$ is the content of the \bbox.
  (Notice that if the \bbox\ is inside
  an exponential box, there could be several stable tokens at its
  principal door, one for each copy of the box.)
\item[(iii)] 
  Tokens cross a sync node $l$ only if for a certain $\bstk$, there is
  a token on each position $(e,\stk,\bstk)$ where $e$ is a premiss of
  $l$, and $\stk$ indicates an occurrence of atom in the type of
  $e$. In this case, all tokens cross the link simultaneously.
  Intuitively, insisting on having the same stack $\bstk$ means that
  the tokens all belong to the same box copy. The simultaneous
  transition of the tokens has to be related to the $s.el$-rule, which
  takes place only when \emph{all} premisses are conclusions of $\onelk$
  nodes. Note that the tokens traverse a sync link only downwards,
  because all edges are positive.
\end{varitemize}

A \emph{run} of the \SIAM{} machine of $R$ is a \emph{maximal}
sequence of transitions $\stI_R \redsiam  \cdots \redsiam
\st_n \redsiam \cdots $ from an initial state $\stI_R$.\\
  
\subsubsection{Basic Properties}
In this and next section, we study some properties of the \SIAM.  We
write $\st \stopsiam$ if no reduction applies from $\st$. A non final
state $\st \stopsiam$ is called a \emph{deadlock} state.  If $\stI_R
\redsiam \st_1 \redsiam ...\redsiam \st_n \stopsiam $ is a run of
$\M_R$ we say that the run \emph{terminates} (in the state $\st_n$).
A run of $\M_R$ \emph{diverges} if it is infinite, \emph{converges}
(resp. \emph{deadlocks}) if it terminates in a final (resp. non final)
state.

\noindent
  \begin{prop}[Confluence and Uniqueness of Normal Forms]\label{lem:diamProp}\label{machine_conf} 
    The relation $\redsiam$ enjoys the following properties:
    \begin{varitemize}
    \item 
      it is confluent and normal forms are unique;
    \item  
      if a run of the machine $\M_R$ terminates, then all runs of $\M_R$ terminate.
    \end{varitemize}
  \end{prop}
  \begin{proof}
    By checking each pair of transition rules we observe that  
    $\redsiam$ has the diamond property, because the transitions 
    do not interfere with each other.
  \end{proof}
  
\condinc{}{\todo By Lemma \ref{lem:diamProp}, all runs of $\M_R$ have the same behaviour. We can therefore say that 
  $\M_R$ \emph{diverges} if    $\M_R$ has an infinite run,  \emph{converges} if its runs converge,   \emph{deadlocks} if its runs do.
By using Lemma \ref{trsf_prop} we can prove that

  \begin{prop}[\todo{What is a good name for this??}]\label{basic_soundness}
    Assume  $\netone \red \nettwo$. $\M_R$ diverges, converges or deadlocks if and only if  $\M_{S}$   diverges, converges or deadlocks.\\

  \end{prop}
}

\subsubsection{Tracing Back}\label{tracing} 
For each position $\pp$ in $R$, we observe (by examining the cases in
Fig.~\ref{fig:trRules}) that there is at most one position from which
$\pp$ can come via a transition. When disregarding the conditions we
impose on rules labelled as (i)--(iii), the transitions also apply to
a single token, in isolation. By reading the transitions
``backwards'', we can therefore define a partial function $\orig :
\POSALL_R \partto \START_R$, where $\orig(\pp) := \ss$ if $\pp$ traces
back to $\ss$. But there is more:
\begin{lemma}\label{lemma:traceback}
  For any state $\st$ such that $\stI_R \redsiam^* \st$, the
  restriction of $\orig $ to $\scod$ is a total, injective function.
\end{lemma}

%

Therefore, for every position $\pp$  which appears in a run of $\M_R$,
$\orig(\pp)$ is defined.

With this in mind, $\START_R$ can be seen as an index set identifying
each token. For most of this section (until Theorem~\ref{soundness})
we are only interested in the ``wave'' of tokens, and do not need to
distinguish them individually. In Section~\ref{sec:beyond}, 
however, we heavily rely on $\orig$ to associate values and
operations to tokens.

\begin{rem}
Tracing back from $\scod$ allows us to reconstruct $\sdom$ from the set of current positions. We have preferred to carry along $\sdom$ in the definition of state to make it more immediate,
since the definition of $\orig$ is rather technical. Similarly, 
 in order not to trace back all the way each time we need the starting position, one can also make the choice to  carry the function  along with the state.  We made a similar choice  in our  
 previous work \cite{lics2014}, where a state was defined as a function $ \sdom \to \POSALL_R$
The two definitions  are of course equivalent for all states which can be reached from the initial state, thanks to Lemma~\ref{lemma:traceback}.
\end{rem}

\subsubsection{State Transformation}\label{def:Transformation}
Our central tool to relate net rewriting and the \SIAM{} is a mapping
of states to states. More precisely, if $\netone \red \nettwo$, we
define a \emph{transformation} as a partial function $\trsf_{\netone
  \red \nettwo}: \POSALL_{\netone} \partto \POSALL_{\nettwo}$, which
extends to a transformation on states $\trsf_{\netone \red \nettwo}:
\sts_{\netone} \partto \sts_{\nettwo}$ in the obvious way,
point-wise. We will omit the subscript $\netone \red \nettwo$ of
$\trsf_{\netone \red \nettwo}$ whenever it is obvious.  

Assume $\netone \red_{a} \nettwo$ (axiom step), and $\pp=(d, s,
\epsilon) \in \POSALL_{\netone}$. If $d\in \{e, f, g\}$ as shown in
Fig.~\ref{fig:trsf}(a), then $\trsf_{\netone \red \nettwo}(\pp):=(h,
s, \epsilon) \in \POSALL_{\nettwo} $.  For the other edges, $
\trsf_{\netone \red \nettwo}(\pp):= \pp$.  This definition can
rigorously be described as in Fig.~\ref{fig:trsf}(b), where the
mapping is shown by the dashed arrows. We give some other cases of
reductions in Fig.~\ref{fig:trsfOthers}. $\trsf$ acts as the identity
on all positions $\pp$ relative to those edges which are not modified
by the reduction rule, \ie, $\trsf (\pp) = \pp$. The cross symbols
$\pmb{\times}$ serves to indicate that the source position has no
corresponding target in $\nettwo$ (remember that the mapping is
partial).  Intuitively, the token on that position is \emph{deleted}
by the mapping.  It is important to observe that in the case of steps
$bot.el$ and $d$ (the only rules which open a box), a stable token is
always deleted.

\begin{fact}\label{fact:trsf}
If $R\red R'$ via a  $bot.el$ or $d$ step, the action of $\trsf$ always delete a stable token.
\end{fact}

\begin{figure}[htbp]
\begin{center}
  \fbox{
    \begin{minipage}{.47\textwidth}
      \begin{center}
      \includegraphics[scale=.6]{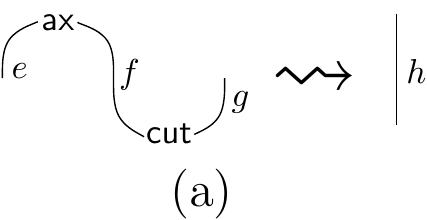}
      \qquad\qquad
      \includegraphics[scale=.6]{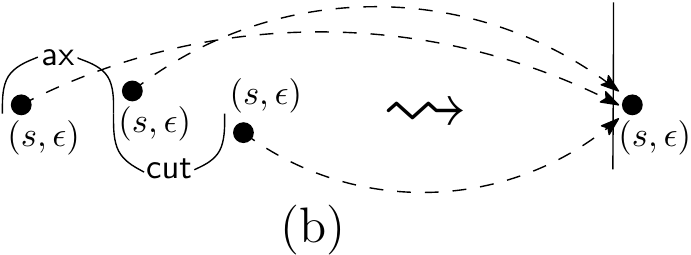}\\
      \end{center}
  \end{minipage}}
\end{center}
  \caption{$\trsf_{R\red S}$, Formally and as a Drawing.}\label{fig:trsf}
\end{figure}
The cases of $d$ and $y$ deserve some further discussion:
\begin{varitemize}
\item   
  In the $d$ rule, the token generated on the $?d$ node is deleted,
  and disappears in $S$. For the other tokens, those outside the !-box
  are modified by removing the signature $*$  
  (which was acquired while crossing that $?d$ node) from the formula stack. 
  The tokens $(e,\stk,\bstk)$ inside the !-box are modified by removing the
  signature $*$ from the \emph{bottom} of the box stack $\bstk$, which
  is coherent with the invariant on the size of $\bstk$ (its size is
  its exponential depth). Why from the bottom of the stack? Because
  the box $b$ which disappears is at depth 0 in $R$, therefore for
  each position $(e,\stk,\bstk)$ inside the box, the signature
  corresponding to $b$ is at the bottom of $\bstk$.
\item 
  In the $y$ rule, things are slightly more complicated. What happens
  to the tokens lying inside a Y-box depends on the bottom element of
  their box stack, which is the signature corresponding to the Y-box.
  If the signature at the bottom of the stack is not of the form
  $y(\cdot,\cdot)$, the token has entered the Y-box only once (\ie, it
  belongs to the first recursive call) and hence the token is mapped
  onto a token in the copy of $S$ outside the Y-box.  Otherwise, the
  token is mapped onto a token in the Y-box; it loses one $y(\cdot,\cdot)$
  symbol (\ie, it does one iteration less), but the box stack becomes
  longer (which is coherent with the increase in depth).  We show an
  example in Fig.~\ref{fig:Ytrsf_example}.  The (stable) token with a
  stack $(\delta, *)$ on the premise of $Y$-node is mapped onto a token
  on the premise of the $\bang$ node, with the same stack.  In
  contrast, the token with a stack $(\delta, y(*,y(*,*)))$ is mapped
  onto a token on a premise of the $Y$ node on the right-hand side,
  now with a stack $(\delta, y(*,*).*)$ --- it loses a $y$ symbol. 
\end{varitemize}

\begin{figure*}[htbp]
  \begin{center}
    \fbox{
      \begin{minipage}{.97\textwidth}
        \begin{center}
        \includegraphics[width=15cm]{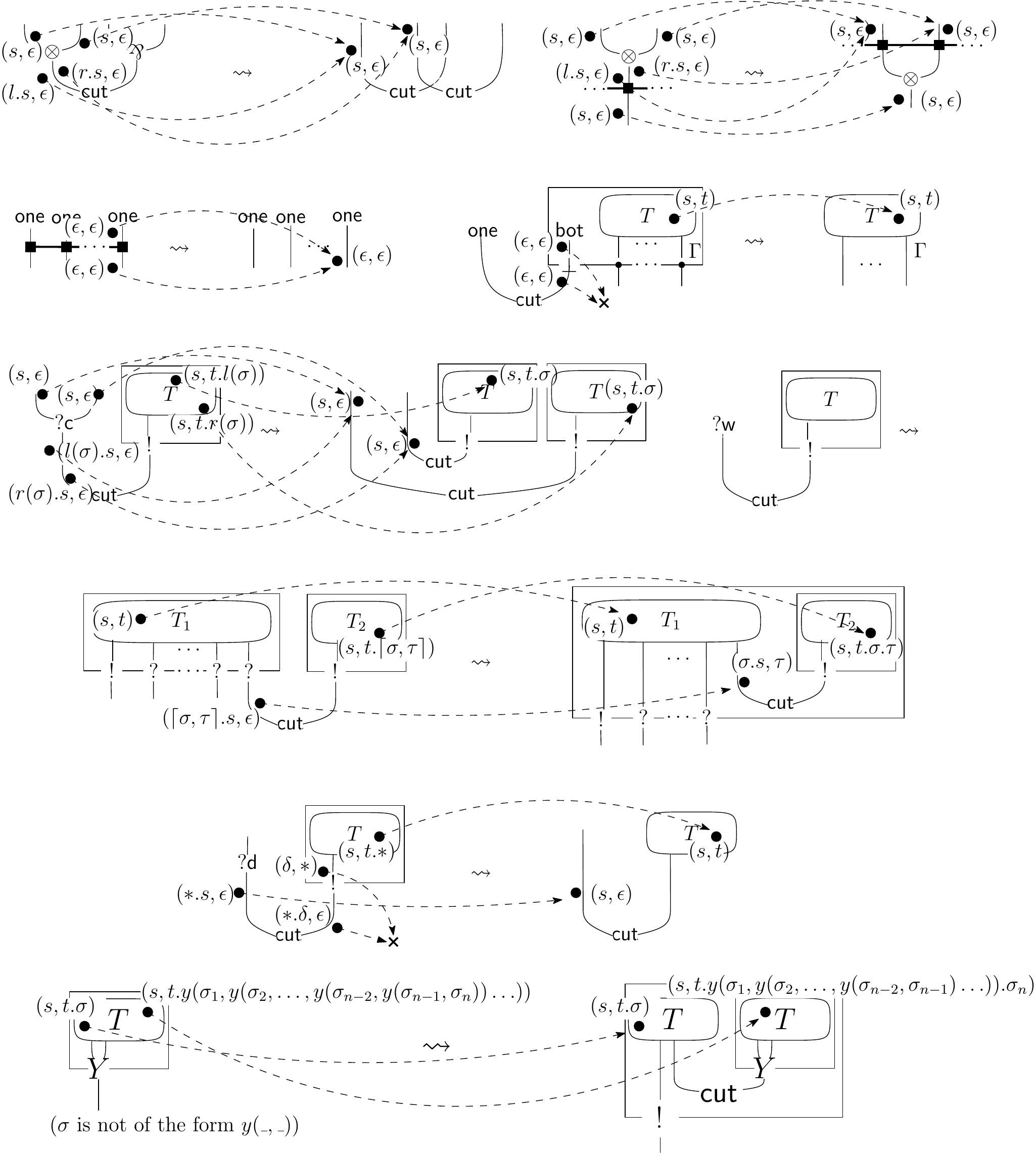}
        \end{center}
    \end{minipage}}
  \end{center}
  \caption{The Function $\trsf_{R\red S}$.}\label{fig:trsfOthers}
\end{figure*}
\begin{figure}[htbp]
  \begin{center}
    \fbox{
      \begin{minipage}{.47\textwidth}
      \begin{center}
        \includegraphics[width=8cm]{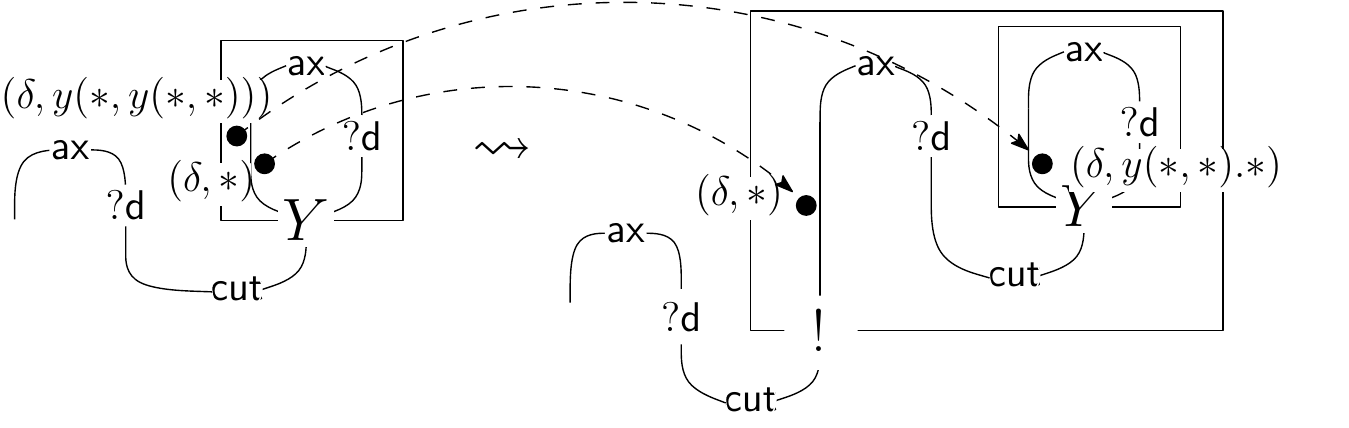}
      \end{center}
      \end{minipage}}
\end{center}
  \caption{$\trsf_{R\red S}$ on $y$-reduction.}\label{fig:Ytrsf_example}
\end{figure}

Each statement below can be proved by case analysis. The proof is given in Appendix \ref{app:SIAM}.
\begin{lemma}[Properties of $\trsf$]\label{trsf_prop}
  Assume $\netone \red \nettwo$.
  \begin{varenumerate}
  \item 
    If $\st \redsiam \sttwo$ in $\machine{\netone}$ then $\trsf(\st)
    \redsiam^* \trsf(\sttwo)$ in $\machine{\nettwo}$.
  \item 
    If $I_R \to \cdots \to \st_n\cdots $ is a run of $\M_R$, then
    $\trsf(I_R) \to^* \cdots \to^* \trsf(\st_n) \cdots$ is a run of
    the machine $\M_{S}$.
  \item\label{trsf_main} 
    $I_R \to \cdots \to \st_n \cdots $ diverges/converges/deadlocks
    iff $\trsf(I_R) \to^* \cdots \to^* \trsf(\st_n) \cdots$ does.
  \end{varenumerate}
\end{lemma}
\newcommand{\der}[1]{\mathtt{weight}(#1)} 
We end this section by looking at the number of circulating tokens.
We observe that the number of tokens, and stable tokens in particular,
in any state $\st$ which is reached in a run of $\M_R$ is finite.  We
denote by $\der{\st}$ the number of stable tokens in $\st$ (\ie,
$\scod \cap \PDOORS_R$). The following is immediate by analyzing
Fig.~\ref{fig:trsfOthers} and checking which tokens are deleted.
\begin{lemma}\label{lem:derTokens}\label{derTokens}
  Assume  $\netone \red \nettwo$. We have that $\der{\st} \geq \der{\trsf(\st)}$.
  Moreover, if $\netone \red \nettwo $  via the $d$-rule or $bot.el$-rule,
  then $\der{\st} > \der{\trsf(\st)}$.
\end{lemma}
\subsection{The Interplay of Nets and Machines}\label{interplay}
We already know that if a simple net $R$ reduces to a normal form $S$, then
$S$ is an \MLL{} net (Corollary~\ref{cutel}), actually a very simple one. It is immediate that  in
this case, every run of the machine $\M_S$ terminates in a final
state: 
each token in the initial state flows to a final position (the net has neither sync nor boxes to stop them). Given an
arbitrary net $R$, we of course do not know if it reduces to a normal
form, but we are still able to use the facts above to prove that
$\M_R$ is deadlock-free.
\begin{lemma}[Mutual Termination]\label{net_termination}
Let $R$ be a simple net, as in Theorem~\ref{main_lem}. We have:
\begin{varenumerate}
\item  
  if a run of $\M_R$ terminates, then each sequence of reductions
  starting from $R$ terminates;
\item 
  if a sequence of reductions starting from $R$ terminates, then each
  run of $\M_R$ terminates in a \emph{final} state.
\end{varenumerate}
\end{lemma}
\newcommand{\ww}[1]{\mathtt{weight}(#1)}
\begin{proof}
Let us first consider Point 1. By hypothesis, there is a run of $\M_R$
which terminates in a state $\st$.  We define $\ww{R}:=\der{\st}$.  By
Lemma~\ref{trsf_prop}, if $R \red S$, $\trsf$ maps the run of $\M_R$
into a run of $\M_{S}$ which terminates in the state $\trsf(\st)$.  By
Lemma \ref{derTokens}, $ \der{\trsf(\st)} \leq \der{\st}$, hence
$\ww{S} \leq \ww{R}$.  Using Lemma \ref{derTokens}, we prove that it
is not possible to have an infinite sequence of $\red$ reductions
starting from $R$, because: (i) each rewriting step which opens a box
($d$, or $bot.el$) strictly decreases $\ww{R}$; (ii) there is only a
finite number of rewriting steps which can be performed without
opening a box. Let us then consider Point 2. By hypothesis, $R$
reduces to a cut free net $S$, which has the form described in Corollary~\ref{cutel}. On such a net, all runs of $\M_S$ terminate in
a final state. If $\M_R$ has a run which is infinite
(resp. deadlocks), by Lemma~\ref{trsf_prop} $\trsf$ would map it into
a run of $\M_S$ which is infinite (resp. deadlocks).
\end{proof}
Lemma~\ref{net_termination} entails deadlock-freeness of the \SIAM{} as an immediate
consequence:
\begin{thm}[Deadlock-Freeness of the \SIAM]\label{SIAM deadlock free}
  Let $R$ be a \SMELLY{} net such that no $?$ appears in its
  conclusions. If a run of $\M_R$ terminates in the state $\st$, then
  $\st$ is a final state.
\end{thm}
\begin{proof}
If $R$ has no $\bot$ and no $!$ in its conclusions, deadlock freeness is immediate consequence of 
Theorem~\ref{net_termination}.
However, the result is true also without this constraint,  because we can always ``close'' the net
$R$ into a net $\overline{R}$ in a way that cannot create any new deadlocks. 

$\overline{R}$ is the net obtained from $R$ when we cut each
  conclusion $A$ of $R$ with the conclusion $A\b$ of the net $S_{A\b}$
  which is defined as follows. $S_{A\b}$ has  the direct encoding of the formula  tree of $A\b$ above the
  conclusion $A\b$ (each modality $?$ is introduced by a $?d$ node);
  the atomic leaves are conclusion of an axiom in the case of $X,X\b,
  \bot$, or of a $\onelk$ node in the case of $\one$. Therefore,  $S_{A\b}$ has only conclusions $X\b,X,1$, \ie\ the other side of the axioms.
  
  To conclude we 
  observe that the \SIAM{} deadlocks 
    in $\overline{R}$ iff it deadlocks in $R$.
\end{proof}

We stress that in the statement above there is \emph{no
  assumption that the conclusions are simple formulas} (unlike in Lemma~\ref{net_termination}, or
Theorem~\ref{main_lem}).
 The constraint that the conclusions are
required not to contain the $?$ modality is instead a real limit,
which is intrinsic to most presentations of GoI (see, \eg,
~\cite{Girard89}).
\subsection{Computational Semantics}\label{semantics}
For the rest of this section, we assume the nets to be simple nets. The reason why  this is not a bothering restriction, is that  
the nets to which we are going to give computational
meaning in the Section~\ref{sec:beyond} are nets where all conclusions have type $1$. 

 The machine $\machine{\netone}$ implicitly gives a semantics to $\netone$.  By
Proposition~\ref{lem:diamProp}, all runs of $\M_R$ have the same
behaviour.  We can therefore say that $\M_R$ either \emph{converges}
(to a unique final state) or \emph{diverges}.  We write
$\M_R\Downarrow$ if all runs of the machine converge. We write
$R\Downarrow$ if all sequences of reductions starting from $R$
terminate in the (unique) normal form. In the previous section we have
established (Lemma~\ref{net_termination}) that
\begin{cor}[Adequacy]\label{thm:termination}
  $\M_R \Downarrow$ if and only if $R \Downarrow$.
\end{cor}
We also already know that:
\begin{cor}[Invariance]\label{basic_soundness}
  Assume $\netone \red \nettwo$.  $\M_R\Downarrow$ if and only if
  $\M_{S}\Downarrow$. 
\end{cor}
We now introduce an equivalence on the machines which is finer than
the one induced by convergence. We associate a partial function $\sem
{R}$ to each net $R$ through the machine $\machine{\netone}$, and show
that $\sem {R}$ is a sound interpretation. This way we have a finer
computational model for $\SMELLY$, on which we will build in the next
sections. The \emph{interpretation} $\sem {\netone}$ of a net
$\netone$ is defined as follows
\begin{varitemize}
  \item  
    if $\machine{\netone}$ diverges, $\sem {\netone} := \Omega$,
  \item  
    if $\machine{\netone}$ converges, $\sem {\netone}$ is the partial
    function $\sem {\netone}: \POSI_{\netone} \partto \POSF_{\netone}$
    where $\sem {\netone}(\ss) := \pp$ if $\pp$ is a final position in
    the final state $\st$ of the machine (\ie, $\pp \in \scod \cap
    \POSF_{\netone}$) and $\orig(\pp) = \ss$.
\end{varitemize}
\begin{theorem}[Soundness]\label{soundness}
  If $\netone \red \nettwo$, $\sem {\netone}= \sem {\nettwo}$.
\end{theorem}
The proof is given in Appendix~\ref{app:SIAM}.
\section{Beyond Nets: Interpreting Programs}\label{sec:beyond}

\SMELLY{} nets as defined and studied in Section~\ref{sec:SMELLY} are
purely ``logical''. In this section we introduce \emph{program nets},
which are a (slight) variation on \SMELLY{} nets in which external
data can be manipulated.  This allows us to interpret \PCF-like
languages. The machine running on these nets will be a very simple
extension of the \SIAM, of which it inherits all properties.
 
The intuition behind program nets is as follows. Assume a language
with a single base type. The base type is mapped to the formula
$\one$; values of the base type are stored in a \emph{memory}.
Elementary operations of the base type are modeled using sync nodes,
recursion is modeled by Y-boxes, conditional tests are captured by a
generalization of the $\bot$-box. Arrow and product types (and all the
usual $\lambda$-calculus constructions) are encoded by means of one of
the well-known mappings of intuitionistic logic into linear
logic~\cite{MaraistOTW95,phdmackie,Girard87}, depending on the chosen
evaluation strategy.

Before introducing program nets and interactive machines for them, let
us fix a language which will also be our main application.
\subsection{\PCF}
The language we shall consider in this section is nothing more than
Plotkin's \PCF, whose \emph{terms} ($ M,N,P$) and \emph{types} ($A,B$) are defined as
  follows:
  
{\footnotesize
\[
\begin{array}{lll}
  M
  &{:}{:}{=} &
  x \midd
  \lambda x.M \midd
  MM \midd
  \PCFlproj(M) \midd
  \PCFrproj(M) \midd\\ &&  
  \PCFpair{M,M} \midd
  \PCFn{n} \midd
  \PCFsucc(M) \midd
  \PCFpred(M) \midd\\ &&
  \PCFifzero{P}{M}{M} \midd
  \PCFletrec{f}{x}{M}{M},
  \\
  A
  &{:}{:}{=} &
  \PCFnat \midd
  A \PCFarrow A \midd
  A \PCFprod A,
\end{array}
\]
}\\
Here, $n$ ranges over the set of non-negative natural
numbers.  A \emph{typing context} $\Delta$ is a (finite) set of typed
variables $\{x_1:A_1,\dots,x_n:A_n\}$, and \emph{typing judgements}
are in the form $\Delta\PCFentail M:A$. We say that a typing judgement
is \emph{valid} if it can be derived from a standard set of typing
rules). Most term constructs are self-explanatory: we
only give a few words on the $\mathtt{letrec}$ construction. In
standard \PCF, the fixpoint is represented with a Y-combinator: while
this is fine in call-by-name evaluation, it does not behave well in the
context of call-by-value reduction. As the $\mathtt{letrec}$ makes
sense in both situations, we use it instead. Moreover, we only want to
allow recursive definitions of \emph{functions}. To syntactically
enforce this, we consider a $\mathtt{letrec}$ binding two variables:
one for the function to be defined, and one for its argument.

A typing context $\Delta$ is a (finite) set of typed variables
$\{x_1:A_1,\dots,x_n:A_n\}$, and a typing judgement is written as
\[
\Delta\PCFentail M:A
\]
A typing judgement is \emph{valid} if it can be derived from the
usual set of typing rules, presented in Table~\ref{tab:typrules}.

On PCF terms, we define a call-by-name and a call-by-value
evaluation, in a standard way.

\begin{table*}
\[
\infer{\Delta,x:A\PCFentail x:A}{}
\qquad
\infer{\Delta\PCFentail \lambda x.M:A\PCFarrow B}{\Delta,x:A\PCFentail
  M:B}
\qquad
\infer{\Delta\PCFentail MN:B}{
  \Delta\PCFentail M:A\PCFarrow B
  &
  \Delta\PCFentail N:A
}
\]
\[
\infer{\Delta\PCFentail \PCFlproj{(M)}:A}{
  \Delta\PCFentail M:A\PCFprod B
}
\qquad
\infer{\Delta\PCFentail \PCFrproj{(M)}:B}{
  \Delta\PCFentail M:A\PCFprod B
}
\qquad
\infer{\Delta\PCFentail \PCFpair{M,N}:A\PCFprod B}{
  \Delta\PCFentail M:A
  &
  \Delta\PCFentail N:B
}
\]
\[
\infer{\Delta\PCFentail\PCFn{n}:\PCFnat}{}
\quad
\infer{\Delta\PCFentail\PCFsucc{(M)}:\PCFnat}{
  \Delta\PCFentail M:\PCFnat
}
\quad
\infer{\Delta\PCFentail\PCFpred{(M)}:\PCFnat}{
  \Delta\PCFentail M:\PCFnat
}
\quad
\infer{\Delta\PCFentail \PCFifzero{P}{M}{N}:A}{
  \Delta\PCFentail P:\PCFnat
  &
  \Delta\PCFentail M:A
  &
  \Delta\PCFentail N:A
}
\]
\[
\qquad
\infer{\Delta\PCFentail \PCFletrec{f}{x}{M}{N}:C}{
  \Delta,f:A\PCFarrow B,x:A\PCFentail M:B
  &
  \Delta,f:A\PCFarrow B\PCFentail N:C
}
\]
\caption{Typing rules for PCF}
\label{tab:typrules}
\end{table*}

\subsubsection{Call-by-name reduction}

A {\em value} in the call-by-name setting is defined from the
following grammar:

{\footnotesize
\[
\begin{array}{lll}
  U
  &
  {:}{:}{=}
  &
  x \midd
  \lambda x.M \midd
  \PCFpair{M,N} \midd
  \PCFn{n}.
\end{array}
\]}
\\
A {\em call-by-name reduction context $C[-]$} is defined with the
following grammar:

{\footnotesize
\[
\begin{array}{lll}
C[-] &{:}{:}{=}
&
[-]\midd
C[-]N\midd \PCFlproj{C[-]}\midd \PCFrproj{C[-]}\midd
\\
&&\PCFsucc(C[-])\midd \PCFpred(C[-])\midd \PCFifzero{C[-]}{M}{N}.
\end{array}
\]}
\\
In call-by-name, $M$ rewrites to $N$, written as $M \PCFcbn N$, is
defined according to the rules presented in Table~\ref{tab:cbnrw}.
\begin{table*}
\begin{center}
\begin{minipage}{.6\textwidth}
(1) {\em Axiom rules.} 
\[
\infer{(\lambda x.M)N \PCFcbn M\{x:=N\}}{}
\qquad
\infer{\PCFlproj{\PCFpair{M,N}} \PCFcbn M}{}
\qquad
\infer{\PCFrproj{\PCFpair{M,N}} \PCFcbn N}{}
\]
\[
\infer{\PCFsucc(\PCFn{n}) \PCFcbn \PCFn{n+1}}{}
\qquad
\infer{\PCFpred(\PCFn{n+1}) \PCFcbn \PCFn{n}}{}
\qquad
\infer{\PCFpred(\PCFn{0}) \PCFcbn \PCFn{0}}{}
\]
\[
\infer{\PCFifzero{\PCFzero}{M}{N}\PCFcbn M}{}
\qquad
\infer{\PCFifzero{\PCFn{n+1}}{M}{N}\PCFcbn N}{}
\]
\[
\infer{\PCFletrec{f}{x}{M}{N} \PCFcbn N\{f := \lambda x.\PCFletrec{f}{x}{M}{f\,x}\}}{}
\]
(2) {\em Congruence rules.} Provided that $C[-]$ is a call-by-name context:
\[
\infer{C[M]\PCFcbn C[N]}{M\PCFcbn N}
\]
\end{minipage}
\end{center}
\caption{Call-by-name reduction strategy for PCF.}
\label{tab:cbnrw}
\end{table*}

\subsubsection{Call-by-value reduction}
A value in the call-by-value setting is defined from the following
grammar

{\footnotesize\[
\begin{array}{lll}
  U
  &
  {:}{:}{=}
  &
  x \midd
  \lambda x.M \midd
  \PCFpair{U,U} \midd
  \PCFn{n}.
\end{array}
\]
}\\
A {\em call-by-value
  reduction context $C[-]$} is defined with the following grammar:
{\footnotesize
\[
\begin{array}{lll}
C[-] &{:}{:}{=}
&
[-]\midd C[-]N\midd VC[-]\midd \PCFpair{C[-],N}\midd
\PCFpair{V,C[-]}\midd
\\
&&
\PCFlproj{C[-]}
\midd \PCFrproj{C[-]}\midd
\PCFsucc(C[-])\midd \PCFpred(C[-])\midd\\
&&\PCFifzero{C[-]}{M}{N}.\phantom{\midd}
\end{array}
\]}

\noindent
In call-by-value, $M$ rewrites to $N$, written as $M \PCFcbv N$, is
defined according to the rules of Table~\ref{tab:cbvrw}.
\begin{table*}
\begin{center}
\begin{minipage}{.6\textwidth}
(1) {\em Axiom rules.} 
\[
\infer{(\lambda x.M)U \PCFcbv M\{x:=U\}}{}
\qquad
\infer{\PCFlproj{\PCFpair{U,V}} \PCFcbv U}{}
\qquad
\infer{\PCFrproj{\PCFpair{U,V}} \PCFcbv V}{}
\]
\[
\infer{\PCFsucc(\PCFn{n}) \PCFcbv \PCFn{n+1}}{}
\qquad
\infer{\PCFpred(\PCFn{n+1}) \PCFcbv \PCFn{n}}{}
\qquad
\infer{\PCFpred(\PCFn{0}) \PCFcbv \PCFn{0}}{}
\]
\[
\infer{\PCFifzero{\PCFzero}{M}{N}\PCFcbv M}{}
\qquad
\infer{\PCFifzero{\PCFn{n+1}}{M}{N}\PCFcbv N}{}
\]
\[
\infer{\PCFletrec{f}{x}{M}{N} \PCFcbv N\{f := \lambda x.\PCFletrec{f}{x}{M}{f\,x}\}}{}
\]
(2) {\em Congruence rules.} Provided that $C[-]$ is a call-by-value context:
\[
\infer{C[M]\PCFcbn C[N]}{M\PCFcbn N}
\]
\end{minipage}
\end{center}
\caption{Call-by-name reduction strategy for PCF.}
\label{tab:cbvrw}
\end{table*}

\subsection{Program Nets and Register Machines}
\newcommand{\surfone}[1]{{\tt SurfOne}(#1)}
\newcommand{\syncnode}[1]{{\tt SyncNode}(#1)}
\newcommand{\syncname}{{\tt SyncNames}}
\newcommand{\oneindex}[1]{{\tt ind}(#1)}
\newcommand{\mapsyncname}[1]{{\tt mkname}(#1)}
\newcommand{\memories}{\mathrm{Mem}} 
\newcommand{\mem}[1]{\mathbf{m}_{#1}} 
\newcommand{\smem}{\mem{\st}} 
\newcommand{\R}{\mathbf{R}}
\renewcommand{\S}{\mathbf{S}}
\newcommand{\test}{\mathrm{test}}
\newcommand{\upd}{\mathrm{update}}
\newcommand{\initop}{\mathrm{init}}
\newcommand{\indset}{\mathrm{I}}
\newcommand{\arity}{\mathrm{arity}}
\newcommand{\nm}{\mathit{l}} 
In the rest of this paper, we assume that all atomic formulas are
units (\ie, $1$ and $\bot$). The language of formulas is therefore $
\formone \bnf \one\midd
\bot\midd\formone\otimes\formone\midd\formone\parr\formone
\midd!\formone\midd ?\formone.$

First of all, we need the
definition of a \emph{memory}:
\begin{deff}\label{def:mem}
  Let $\indset$ be a (possibly) infinite set whose elements are called
  \emph{addresses}. Let $\syncname$ be a finite set of names, where to
  each name we associate a positive number that we call {\em
    arity}. Given a set of \emph{values} $\mathbb{X}$, we define
  $\memories $ as the set $\indset\to\mathbb{X}$ of all functions from
  $\indset$ to $\mathbb{X}$, equipped with the following operations:
  $$
  \begin{array}{r@{{~~}:{~~}}l}
    \test & \indset\times\memories \to \mathrm{Bool}\times\memories;
    \\
    \upd & 
    \syncname\times\displaystyle(\indset^*)\times\memories
    \rightharpoonup \memories;
    \\
    \initop & \indset\times\memories \to \memories.
  \end{array}
  $$
  where the  partial function $\upd$ is defined on a triple $(\nm,x,\mem{})$ iff
      the length of $x$ equals the arity of $\nm$.
      
  A \emph{memory}\footnote{An even more fitting name would be
    \emph{memory states}, but we do not want to overload too much the
    term ``state''.} is any element $\mem{}$ of $\memories$, and we say
  that $\mem{}$ \emph{has values in} $\mathbb{X}$.
\end{deff}
Intuitively, $\mem{}$ represents a set of \emph{registers} which are
referenced by the elements of $\indset$ (the addresses). The operation
$\test$ is used to query the value of a register, $\upd$ to update its
value, and $\initop$ to set a register of the memory to a default
value. Some comments on the operations on $\memories$ are useful. The
reason why we have $\memories$ in the codomain of the operation
$\test$, is that we aim at a general model where $\test$ might have a
non-local effect on the memory, such as in a quantum setting (see
e.g.~\cite{esop}), though its implementation is beyond the scope of
this paper.  Notice also that the type of $\upd$ is really a
dependent-type.
\subsubsection{Program Nets}
Program nets are obtained as a light and natural extension of
\SMELLY{} nets, as follows:
\begin{varitemize}
\item   
  \bboxes\ are replaced by multi-\bboxes, which are meant to handle
  \emph{tests}. 
   A
  \emph{multi-$\bot$-box} is a $\bot$ node to which we associate
  \emph{two} structures with the same conclusions $\Gamma$, as shown
  in Fig.~\ref{multibox} and~\ref{SIAMmultibox} (these figures are
  fully explained later on). 
  An {\em extended {\SMELLY} net} is a {\SMELLY} net where
  multi-\bbox{}es\footnote{In some example pictures, it is still
    convenient to use simple \bboxes; they can be seen as a short-cut
    for multi-\bboxes\ with the same net in both places.  } are used in
  place of \bbox{}es.
\item 
  Given an (extended) net $R$, let $\surfone{R}$ be the set of all
  $\onelk$ nodes at the surface, and $\syncnode{R}$ be the set of
  {\em all} sync-nodes of the extended net $R$, whether at surface
   or not.  A {\em decoration} of $R$ with names $\syncname$
  consists of the following two pieces of data:
  \begin{varenumerate}
  \item 
    An injective partial map $\oneindex{R}:\surfone{R}\partto
    \indset$, \ie, $\onelk$ nodes are not necessarily decorated;
  \item 
    A {\em total} map $\mapsyncname{R}:\syncnode{R}\to \syncname$,
    where $\syncname$ is a finite set of names. This map is simply
    naming the sync nodes appearing in the extended net $R$. We assume
    that given a name of arity $k$, all the sync nodes decorated with
    that name have the same arity $k$, where the arity of a sync node
    is the total number of $1$'s in its premisses.
  \end{varenumerate}
\end{varitemize}
\begin{deff}
   Given a  set $\memories$ as in Definition~\ref{def:mem}, 
  a {\em program net} is a pair $\R= (R,\mem{R})$, where $R$ is a
  decorated, extended net and $\mem{R}\in \memories$ is a memory.
\end{deff}
Rewriting on {\SMELLY} nets easily extends to program nets as shown in
Fig.~\ref{multibox} (where we adopt the convention that the memory
associated to the net is $\mem{1}$ before reduction, and $\mem{2}$
after reduction). The rules are as follows. Rule ${\it decor}$ is a
new rewriting rule which associates to a surface node $\onelk$ an
address $r\in I$; when doing this, we are \emph{linking} the $\onelk$
node to the memory. Rule ${\it bot.el}$ is modified to reflect the use
of multi-$\bot$-boxes.  As shown in Fig.~\ref{multibox}, the reduction
depends on the memory, and is determined by the result of the
operation $\test$. For the other reduction rules, the underlying net
is rewritten exactly as for \SMELLY\ nets.  Concerning the memory,
only the rule ${\it s.el}$ modifies it, as follows:
$\mem{2}=\upd(\nm,(r_1,r_2,\dots r_k), \mem{1})$, where $k$ is the
arity of $\nm$. In all the remaining cases $\mem{1}=\mem{2}$ (i.e. the
memory is not changed)
\begin{figure}[htbp]
  \begin{center}
    \fbox{
      \begin{minipage}{.47\textwidth}
        \begin{center}
          \includegraphics[width=8.2cm]{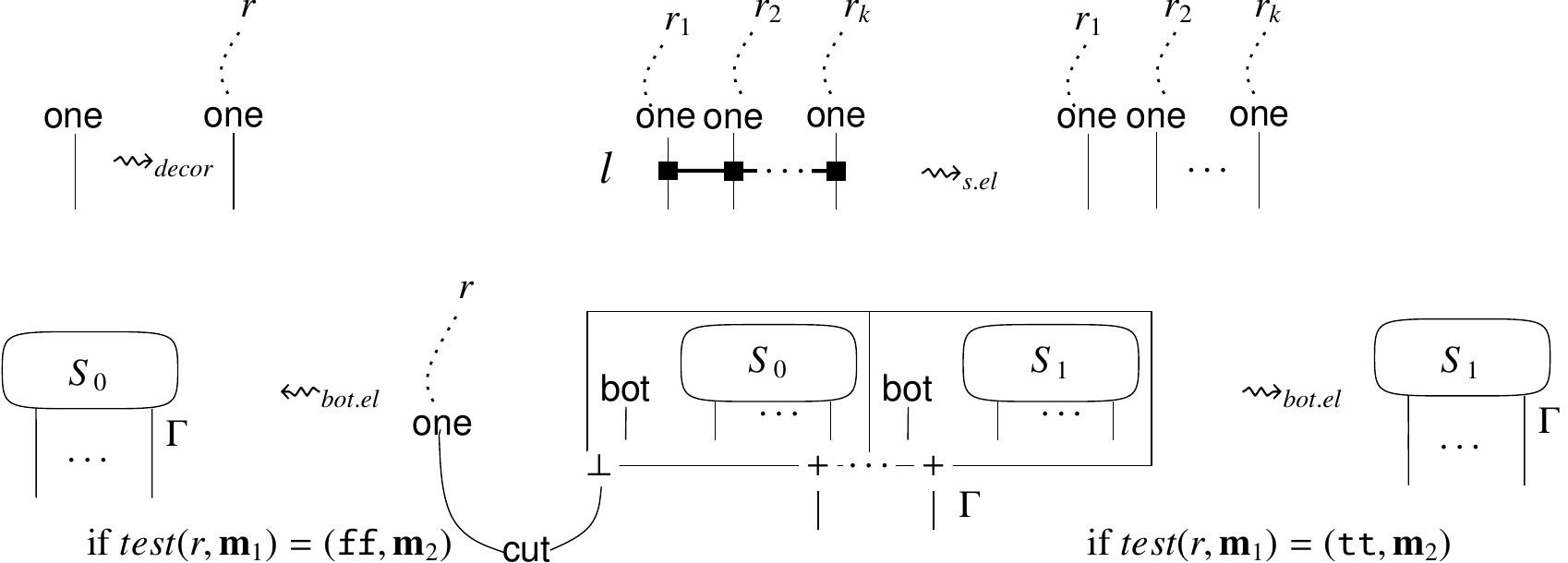}
        \end{center}
      \end{minipage}}
\end{center}
  \caption{Program Net Rewriting.}\label{multibox}
\end{figure}
What we have introduced so far is a general schema for program nets;
in order to capture specific properties, we need to define $\memories$
and the operations on it.  In the following section, we specialize the
construction to \PCF.
\subsubsection{\PCF\ nets}\label{pcf nets}
\newcommand{\ttrue}{\mathtt{tt}} \newcommand{\tfalse}{\mathtt{ff}} To
encode \PCF\ programs, we use a class of program nets.  Once
$\memories$ and the operations on it are appropriately defined, we are
able to gain more expressive power than in \SMELLY{}, while good
computational properties will be still guaranteed by the underlying
nets. A {\em \PCF\ net} is a program net where
$\memories$ has values in
$\mathbb{N}$, that is, $\memories ~{:}{=}~ \indset\to\mathbb{N}$.  The
set of sync-names is $\{\mathtt{max}, \mathtt{p}, \mathtt{s}\}$:
$\mathtt{max}$ is binary while $\mathtt{p}$ and $\mathtt{s}$ are
unary.  The operation $\upd$ is defined as follows.  { The sync node
  of label $\mathtt{p}$ acts as the predecessor, that is $\upd(\mathtt{è},r,\mem{1}) =
  \mem{2}$ where $\mem{2}(r)=\mem{1}(r)-1$ and $\mem{2}(k)=\mem{1}(k)$
  if $k\neq r$.  The node of label $\mathtt{s}$ acts as the successor, that is
  $\upd(\mathtt{s},n,\mem{1}) = \mem{2}$ where $\mem{2}(r)=\mem{1}(r)+1$ and
  $\mem{2}(k)=\mem{1}(k)$ if $k\neq r$.  Finally, the sync node of
  label $\mathtt{max}$ acts as follows: $\upd(\mathtt{max},r,q,\mem{1}) =
  \mem{2}$ where $\mem{2}(r)=\mem{2}(q)=\max(\mem{1}(r),\mem{1}(q))$
  and $\mem{2}(k)=\mem{1}(k)$ if $k\neq r$ and $k\neq q$.}  For the
other operations, $\test(r,\mem{})$ is defined to be $(\ttrue,
\mem{})$ if $\mem{}(r) = 0$, and $(\tfalse, \mem{})$ otherwise;
$\initop(r,\mem{})$ is defined to be the memory $\mathbf{n}$ where
$\mathbf{n}(r) = 0$ and $\mathbf{n}(k) = \mem{}(k)$ for $k \neq r$.
Any typing derivation is encoded as a \PCF\ net.  Two possible
encodings will be considered: one for call-by-value, one for
call-by-name, which correspond to two translations of intuitionistic
logic into linear logic~\cite{MaraistOTW95,phdmackie}.

\subsubsection{Register Machines}  
\begin{figure}
\begin{center}
  \fbox{
  \begin{minipage}{.47\textwidth}
    \begin{center}
    \includegraphics[width=7.2cm]{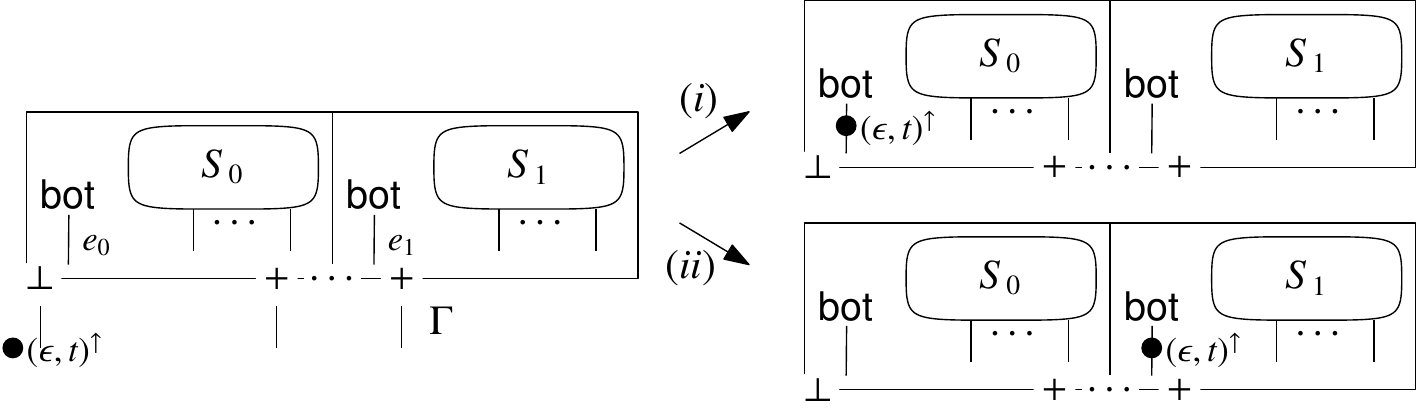}
    \end{center}
  \end{minipage}}
\end{center}
\caption{Multi-\bbox\ Transition for a Register Machine.}\label{SIAMmultibox}
\end{figure}
The \SIAM{}, as we defined it in Section~\ref{SIAM}, is readily
adapted to interpret \PCF{} nets. Let us first sketch a general
construction for the machine which is associated to a program net.
The dynamics of the machine is mostly inherited from the \SIAM; the
novelty is that the notion of state now includes a {memory}.  Let us
fix a set of memories $\memories{}$.  To a program net $\R= (R,
\mem{R})$ ($\mem{R}\in \memories$) is associated the machine $\M_{\R}$
whose memories, states and transitions are defined as follows. The
definition of position and set of positions is the same as in
Section~\ref{SIAM}.
\paragraph*{Memories} 
$\memories$ and the operations on it are the same as for the program
net $\R$. To illustrate the machine, we need however to make the
set of addresses $\indset$ precise. We take $\indset$ to be the set of
positions $\POSI_R \cup \ones_R$. We say that the access to the memory
is defined for all positions for which $\orig (\pp) \in \POSI_R
\cup \ones_R$.
\paragraph*{States}
A state of $\M_{\R}$ is a pair $(\st, \smem)$, where $\st$ is a state
in the sense of Section~\ref{SIAM}, and $\smem\in \memories{}$ is a
memory. An \emph{initial} state of $\M_{\R}$ is a pair $(\stI,
\mem{\stI})$, where $\mem{\stI}$ coincides with $\mem{R}$ for the
positions corresponding to decorated $\onelk$ nodes, is arbitrary on
$\POSI_R$, and is $0$ anywhere else.
\paragraph*{Transitions}
The transitions are the same as in \ref{SIAM}, except in the following
cases, which are defined only if the access to the memory is defined.
\begin{varitemize}
\item 
  Sync nodes. When the tokens $\pp_1, \dots, \pp_k$ cross a sync node
  with label $\nm$ and arity $k$, the operation $\upd(\nm,\pp_1,
  \dots, \pp_k, \mem{})$ opportunely modifies the memory $\mem{}$.
\item 
  Multi-$\bot$-box. Let the box be as in Fig.~\ref{SIAMmultibox},
  where $S_0$ and $S_1$ are the two nets associated to it, and the
  edges $e_0,e_1$ are as indicated.  When a token is in position
  $\pp=(e,\emp,\bstk)$ on the principal conclusion of the box, it
  moves to $ (e_0,\emp,\bstk)$ if $\test(\orig(\pp),\smem)$ returns
  the boolean $\mathtt{ff}$ (arrow (i) in Fig.~\ref{SIAMmultibox}) and
  it moves to $(e_1,\emp,\bstk)$ if $\test(\orig(\pp),\smem)$ returns
  $\mathtt{tt}$ (arrow (ii) in Fig.~~\ref{SIAMmultibox}). If a token
  $(f, \stk, \bstk)$ is on an auxiliary conclusion $f$, it moves to
  the corresponding conclusion in $S_0$ (resp. $S_1$) if $\bstk\in
  \id(S_0)$ (resp. $\bstk \in \id(S_1)$).
\end{varitemize}
\paragraph*{State Transformations}
Let $\R=(R,\mem{R})$ be a program net and $\R \red \S=(S, \mem{S})$.
The transformation $\trsf$ described in Fig.~\ref{fig:trsfOthers}
associates \emph{positions} of $R$ to \emph{positions} of $S$; this
allows us also to specify the transformation of the memory, hence
allowing us to map a memory of $\M_\R$ into a memory of $\M_{\S}$.
More precisely, each state $(\st, \smem )$ of $\M_\R$ is mapped 
into a state $(\trsf(\st), \trsf(\smem) )$ of  $\M_{\S}$.
\subsubsection{\PCF\ Machines}
A \PCF\ machine is a {register} machine where $\memories$ and the
operations on it are defined as for \PCF\ nets (Section \ref{pcf
  nets}).  As for the \SIAM, we have that $\trsf$ maps each run of
$\M_\R$ into a run of $\M_{\R'}$ which converges/diverges/deadlocks
iff the run on $\M_R$ does.  By combining \PCF\ nets and the
\PCF\ machine, it is possible to establish similar results to those in
Section~\ref{interplay} and~\ref{semantics}.  \condinc{}{
\begin{lemma} Let $\R$ be a \PCF\ net where all conclusions have type $\one$.
  The machine $\M_\R$ terminates in a final state (say $(\st,\smem)$)
  iff $\R$ reduces to a cut and sync free net (say $\mathbf S=(S, \mem{S})$).
  Moreover 
  \[ 
  \mem{S}=\sem {\smem}
  \] 
  
  where $\sem{\blue{\smem}}$ is the restriction of $\blue{\smem}$
  to the elements pointed to by final positions.
\end{lemma}
}
Assume $\R$ is a \PCF\ net of conclusion $\one$. We write $\R
\Downarrow n$ if $\R$ reduces to $\mathbf S$, where the value in the
memory corresponding to the unique $\onelk$ node in $\mathbf S$ is
$n$. Similarly we write  $\M_{\R} \Downarrow n$, where $n$ is the value pointed to by the unique final position in the final state of $\M_{\R}$.
\begin{thm}[Adequacy]
$\R \Downarrow n$ if and only if $\M_{\R} \Downarrow n$.
\end{thm}

\subsection{The Call-by-Value Encoding}
In the call-by-value encoding of \PCF\ into \PCF\ nets, the shape of
the net corresponding to $x_1:A_1,\dots,x_n:A_n\PCFentail M:B$ is
\[
\includegraphics[page=2,width=5cm]{PCFcbv.pdf}
\]
where $M^\dagger$ is a net and where $(\cdot)^\dagger$ is a mapping of types to
{\SMELLY} formulas, defined as follows:
\begin{align*}  \PCFnat^\dagger&:=~\one;\\
  (A\PCFarrow B)^\dagger&:=~{!(A^\dagger\b\parr B^\dagger)};\\
  (A\PCFprod B)^\dagger&:=~{A^\dagger}\tens{B^\dagger}.
\end{align*}

\noindent
In our translation, we have chosen to adopt an \emph{efficient}
encoding, rather than the usual call-by-value encoding. In other
words, we follow Girard's optimized translation of intuitionistic into
linear logic, which relies on properties of positive formulas
\cite{Girard87}\footnote{A good summary of the different translations
  is given at the address
  \url{http://llwiki.ens-lyon.fr/mediawiki/index.php/Translations_of_intuitionistic_logic}}.
We feel that this encoding is closer to call-by-value computation than
the non-efficient one; it however raises a small issue.  Notice in
fact that we map natural numbers into the type $1$, not $!1$. How
about duplication and erasure, then? We will handle this in the next
section, by using sync nodes, but let us first better clarify what the
issue is.

Girard's translation relies on the fact that $1$ and $!1$ are
logically equivalent (\ie, they are equivalent for
provability). However, this in itself is not enough to capture
duplication in our setting, because we need to also duplicate the
values in the memory, and not only the underlying net. We illustrate
this in Fig.~\ref{one_equiv}. The portion inside the dashed line
corresponds to a proof of $1\vdash {!1} $; when we look at an example
of its use (l.h.s. of the figure), we see that by using it we do
duplicate the node $\onelk$, but \emph{not} the value $n$ which is
associated to it. The value $n$ is not transmitted from the $1$ to the
$!1$ which is going to be duplicated. The logical encoding however
still correctly models weakening (r.h.s. of Fig.~\ref{one_equiv}).
\begin{figure}
\begin{center}
\fbox{
  \begin{minipage}{.47\textwidth}
    \begin{center}
    \includegraphics[width=7cm]{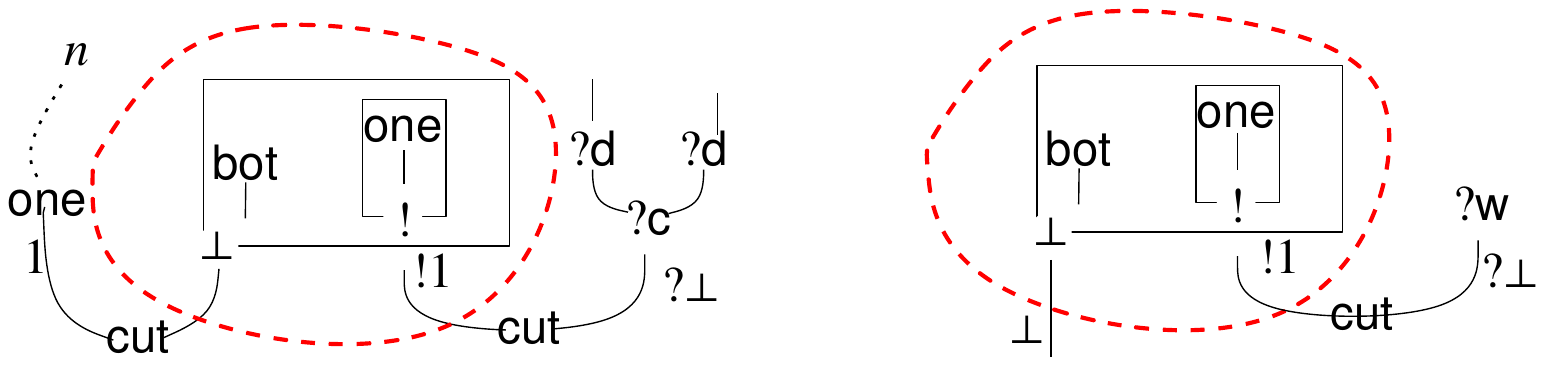}
    \end{center}
  \end{minipage}}
\end{center}
\caption{A Proof of $1\vdash {!1}$.}\label{one_equiv}
\end{figure}
\paragraph*{Exponential Rules and the Units}
The formula $\bot$ does not support contraction, weakening and
promotion ``out of the box'' in {\SMELLY} but it is nonetheless
possible to encode them as \PCF\ nets with the help of the binary sync
node {\tt max}.
\begin{figure}
\begin{center}
\fbox{
\begin{minipage}{.47\textwidth}
  \begin{center}
    \includegraphics[page=4,width=4cm]{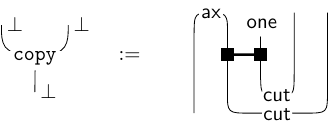}
  \end{center}
\end{minipage}}
\end{center}
\caption{Syntactic Sugar: Copying $\bot$.}\label{fig:copy-node}
\end{figure}
\begin{varitemize}
\item 
  {\em Contraction.} We encode contraction on $\bot$ by using a sync
  node {\tt max} and the syntactic sugar {\tt copy} defined in
  Fig.~\ref{fig:copy-node}. It duplicates the value associated to the
  incoming edge, and it does so in a call-by-value manner: it will
  only copy a $\onelk$ node (i.e. a result), not a whole
  computation. In particular, it should be noted that the rules of net
  rewriting are not modified.
\item 
  {\em Promotion.} We aim at the reduction(s) shown in
  Fig.~\ref{fig:prom-one-idea}: a $\onelk$ node with memory set to $n$
  is sent to a frozen computation (inside a $!$-box) computing the
  same $\onelk$ node. Since {\SMELLY} features recursion in the form
  of the $Y$-box, together with the copy operation already defined it
  is possible to write a net for the formula $\bot\parr{!\one}$, as
  shown in Fig.~\ref{fig:prom-one-net}.
\item 
  {\em Weakening.}  We can directly use the coding given on the
  r.h.s. of
  Fig.~\ref{one_equiv}.\\
\end{varitemize}
\begin{figure}
\begin{center}
\fbox{
  \begin{minipage}{.47\textwidth}
    \begin{center}
      \includegraphics[page=6,width=8.2cm]{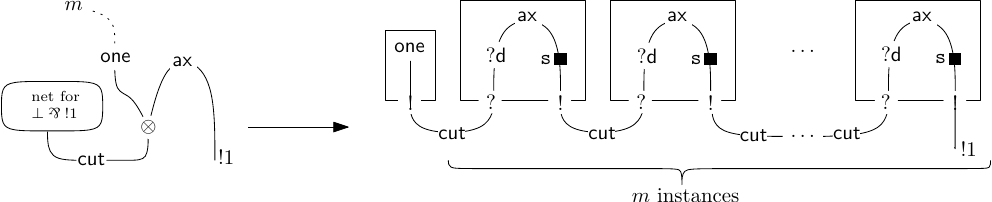}
    \end{center}
  \end{minipage}}
\end{center}
\caption{Desired Behavior for the Mapping of $\one$ to ${!\one}$.}\label{fig:prom-one-idea}
\end{figure}
\begin{figure}
\begin{center}
\fbox{
  \begin{minipage}{.47\textwidth}
    \begin{center}
      \includegraphics[page=7,width=6cm]{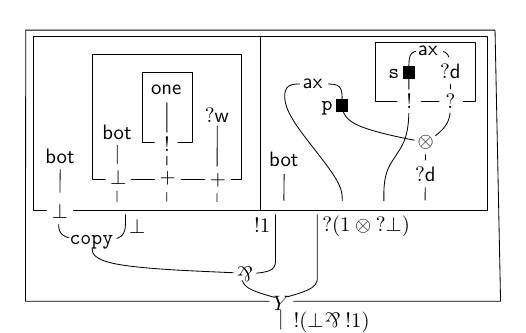}
    \end{center}
  \end{minipage}}
\end{center}
\caption{\PCF\ Net Computing $\bot\! \parr {!1}$.}\label{fig:prom-one-net}
\end{figure}
\paragraph*{Exponential Rules for the Image $A^\dagger$ of any Type $A$}
The goal of this paragraph is to construct nets which behave like the
nodes $?c$, $?w$ and $?p$ of linear logic, this for any edge of type
$A^\dagger\b$.  For any type $A$, the formula $A^\dagger$ is a
multi-tensor of $\one$'s and $!$-ed types. We therefore construct the
grey contraction, weakening and promotion nodes inductively on the
structure of the type, as presented in Fig.~\ref{fig:cwp-pos-type}.
\begin{figure*}
\begin{center}
  \fbox{
    \begin{minipage}{14cm}
    \begin{center}
      \includegraphics[width=13cm,page=1]{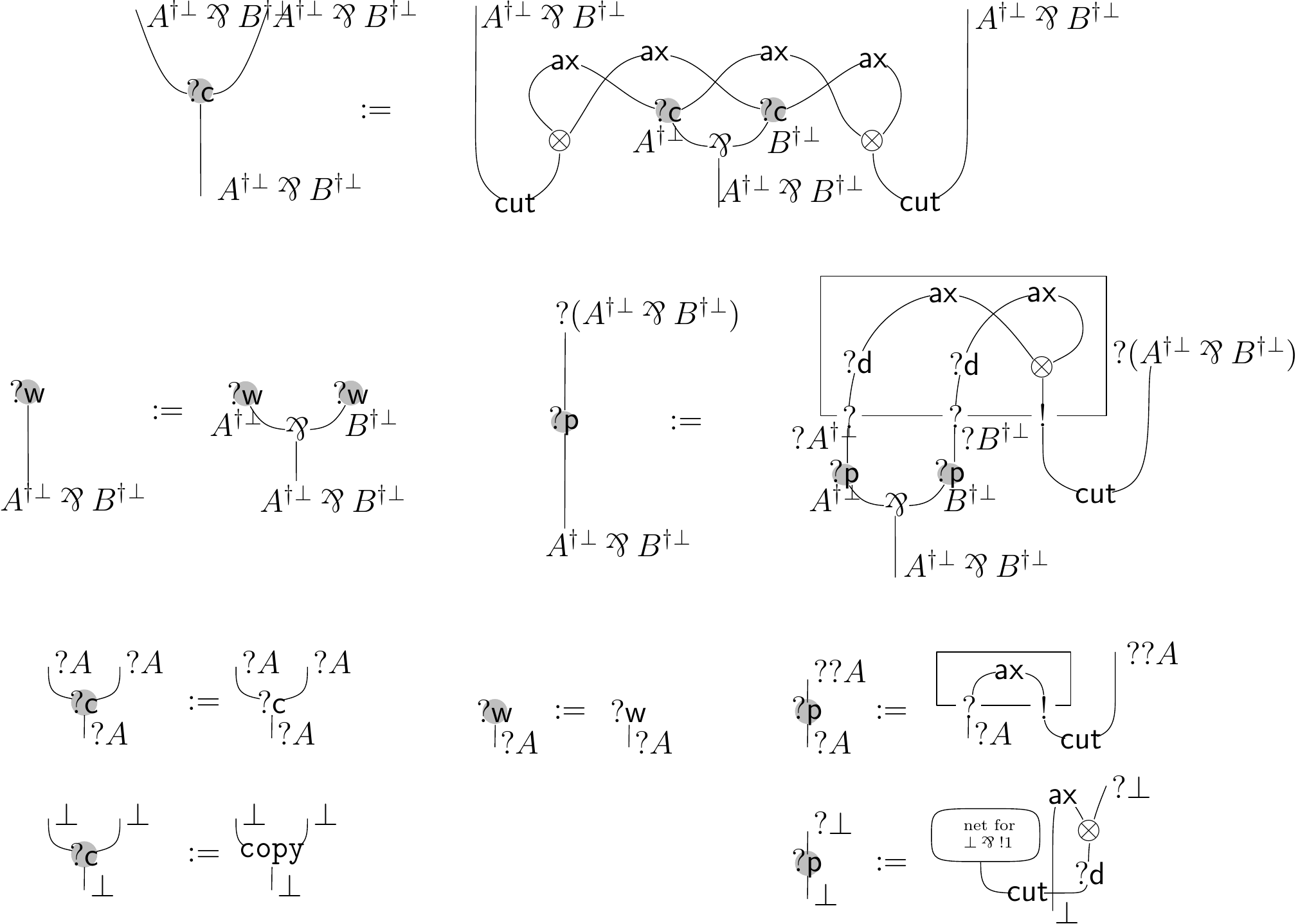}
    \end{center}
  \end{minipage}}
\end{center}
\caption{Inductive Definition of Contraction, Weakening and Promotion
  Nodes.}\label{fig:cwp-pos-type}
\end{figure*}
\paragraph*{Interpreting Typing Judgements}
Typing derivations are inductively mapped to \PCF\ nets as shown in
Fig.~\ref{fig:PCFcbv}. The grey nodes $\mathsf{?c}$ and $\mathsf{?w}$
were defined in Fig.~\ref{fig:cwp-pos-type} (the case of $?\bot$ has
been discussed above). The grey node ``$\mathsf{?}$'' is a shortcut
for the following construction:
\begin{center}
  \includegraphics[page=3,height=2cm]{PCFcbv-positive-nodes.pdf}
\end{center}
\begin{figure*}
\begin{center}
  \fbox{
    \begin{minipage}{15cm}
      \begin{center}
        \includegraphics[page=1,width=14cm]{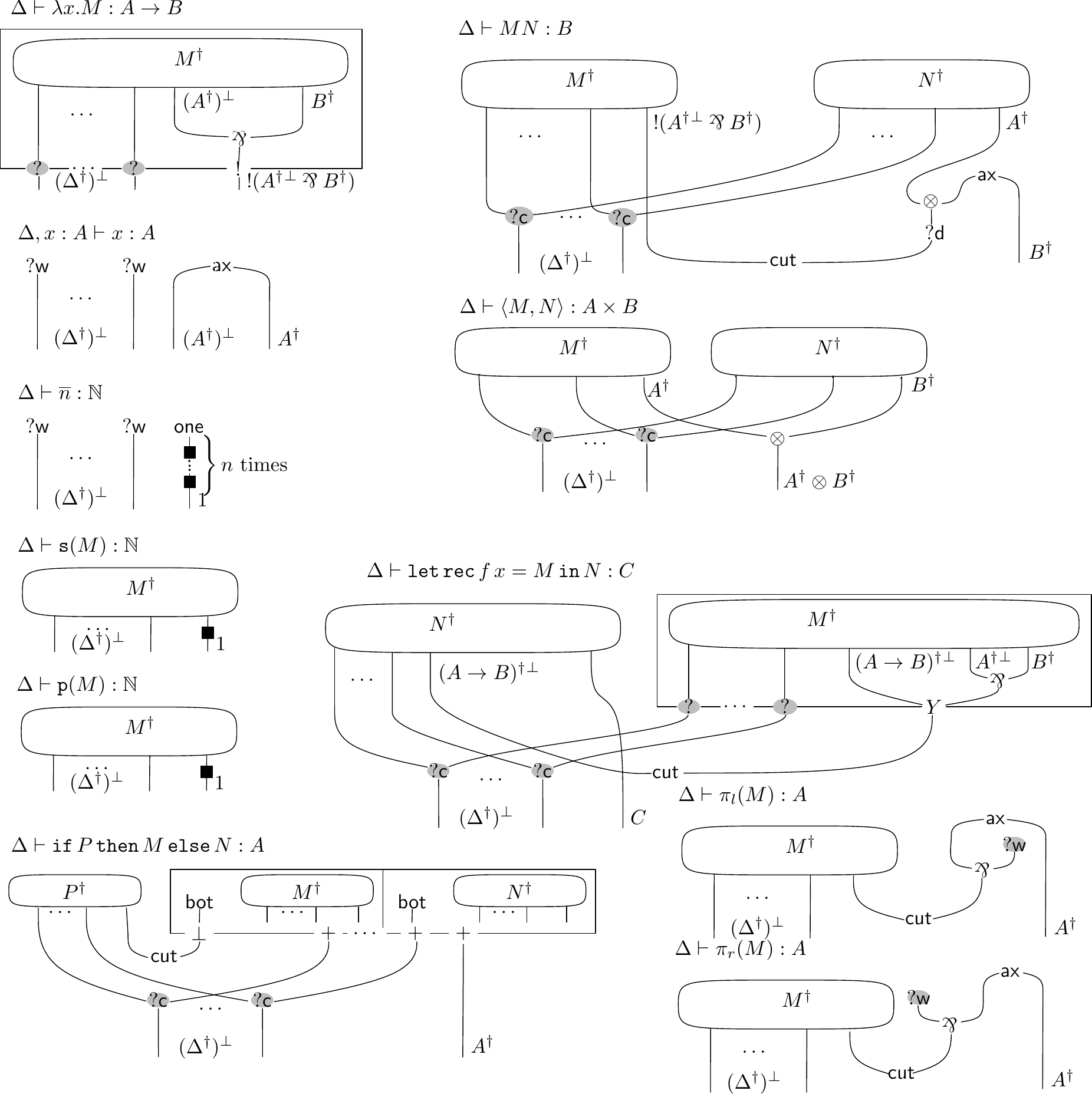}
      \end{center}
    \end{minipage}}
\end{center}
\caption{Call-by-Value Translation of \PCF\ into \PCF\ Nets.}\label{fig:PCFcbv}
\end{figure*}
\paragraph*{Adequacy}
We prove the following result, which relates the call-by-value
encoding into \PCF\ nets and the call-by-value reduction strategy for
terms:
\begin{theorem}
  \label{th:adequacy-cbv}
  Let $M$ be a closed term of type $\PCFnat$. Then $M\PCFcbv\PCFn{n}$
  if and only if $M^\dagger\Downarrow n$.
\end{theorem}
As a corollary, we conclude that the machine on $M^\dagger$ behaves as
$M$ in call-by-value.
\begin{corollary}
  Let $M$ be a closed term of type $\PCFnat$. Then $M$ call-by-value
  converges if and only if $\M_{M^\dagger}$ itself converges.
\end{corollary}

\subsection{The Call-by-Name Encoding}
Besides the encoding of call-by-value \PCF, which is non-standard, and
has thus been described in detail, program nets also have the
expressive power to encode call-by-\emph{name} \PCF. The encoding is
the usual one: a proof net corresponding to
$x_1:A_1,\dots,x_n:A_n\PCFentail M:B$ has conclusions
$\{?A_1^*\b,\ldots, ?A_n^*\b, B^*\}$
\[
\includegraphics[page=2,width=5cm]{PCFcbn.pdf}
\]
where $(\cdot)^*$ is a mapping
of types to {\SMELLY} formulas:
\begin{align*}
\PCFnat^*&{}:=~\one;\\
  (A\PCFarrow B)^* &{}:=~ ?(A^*)\b\parr B^*;\\
  (A\PCFprod B)^* &{}:=~ {!(A^*)}\tens{!(B^*)}.
\end{align*}
\noindent Typing derivations are mapped to \PCF{} nets essentially in
the standard way, and presented in Figure~\ref{fig:PCFcbn}. Note that
unlike the call-by-value translation, since every context is always in
$?$-form, we do not need special weakening, contraction and promotion
nodes.
Then, as in the previous section, one can relate the
call-by-name encoding in \PCF\ nets and the call-by-name reduction
strategy for terms.

\begin{figure*}
\begin{center}
  \fbox{
    \begin{minipage}{15cm}
      \begin{center}
        \includegraphics[page=1,width=14cm]{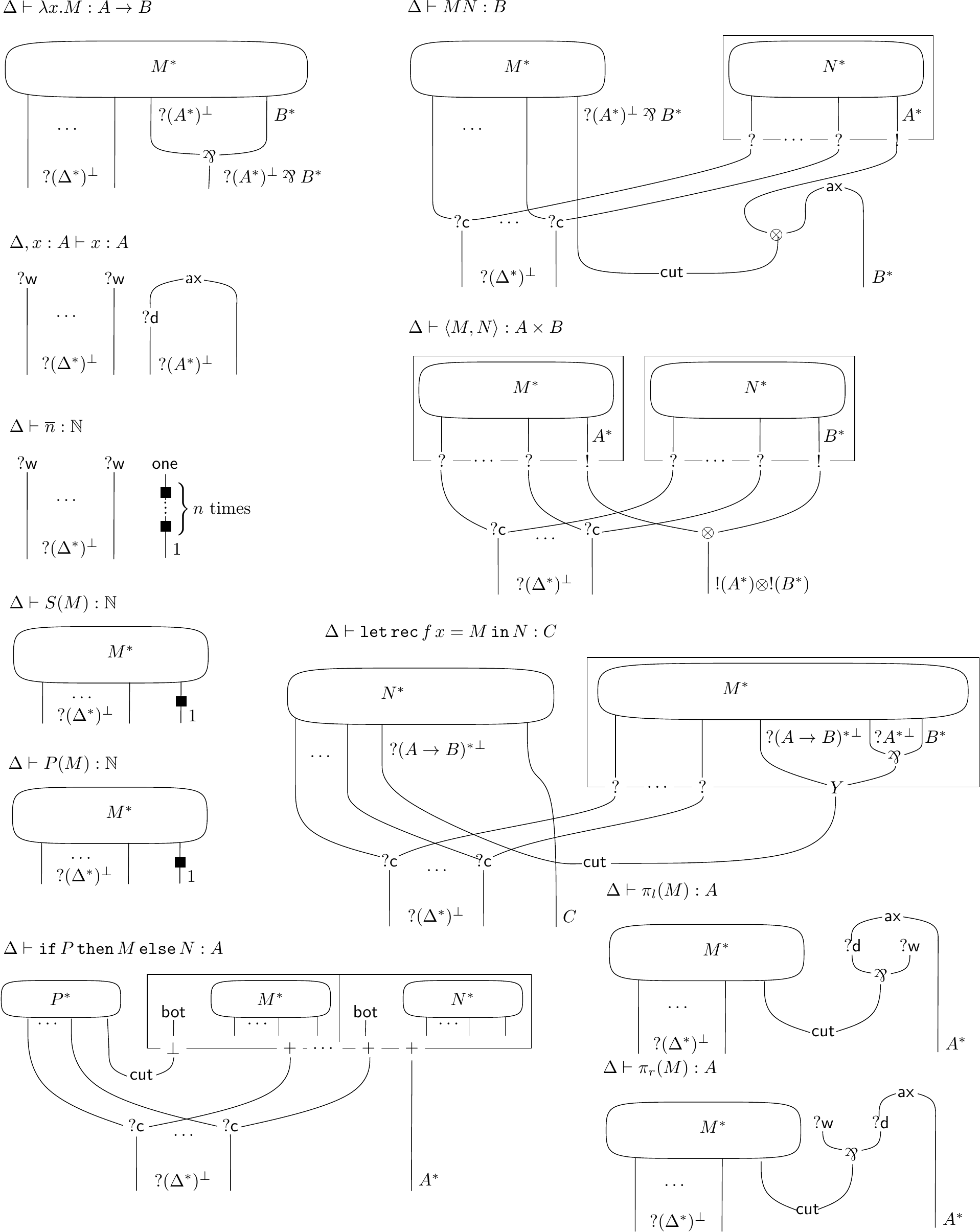}
      \end{center}
    \end{minipage}}
\end{center}
\caption{Call-by-Name Translation of \PCF\ into \PCF\ Nets.}\label{fig:PCFcbn}
\end{figure*}

\begin{theorem}[Adequacy]
  \label{th:adequacy-cbn}
  Let $M$ be a closed term of type $\PCFnat$. Then $M\PCFcbn\PCFn{n}$
  if and only if $M^*\Downarrow n$.
\end{theorem}
As a corollary, one can show that the  machine on $M^*$ behaves
as $M$ in call-by-name.
\begin{corollary}
  Let $M$ be a closed term of type $\PCFnat$. Then $M$ converges in
  call-by-name if and only if the register machine $\M_{M^*}$ itself converges.
\end{corollary}
\section{Conclusions}
We have shown how the multitoken paradigm not only works well in the
presence of exponential and fixpoints, but also allows us to treat
different evaluation strategies in a uniform way. Some other
interesting aspects which emerged along the last section are worth
being mentioned.

In the call-by-value encoding of \PCF, we have used binary \emph{sync
  nodes} in an essential way, to duplicate values in the register:
without them, the efficient encoding of natural numbers would not have
been possible. This shows that sync nodes can indeed have an
interesting computational role besides reflecting entanglement in
quantum computation~\cite{lics2014}. In the future, we plan to further
the potential of such an use, in particular in view of
efficient implementations.

A key feature of \SMELLY{} nets rewriting is that it is
\emph{surface}. Surface reduction allows us to interpret recursion,
but how much do we lose by considering surface reduction instead of
usual cut-elimination? We think that a simple way to understand the
limitations of surface reduction is to consider an analogy to
Plotkin's weak reduction. In \PCF, $\lambda x.\Omega$ is a normal
form. As a consequence one loses, \eg, some nice results about the
shape of normal forms in the $\lambda$-calculus (which, in logic,
corresponds to the subformula property). In presence of fixpoints,
however, this is a necessary price to pay. Otherwise, any term
including a fixpoint would diverge. Of course there is much more to be
said about all this, and we refer the reader to, \eg, the work by
Simpson~\cite{SimpsonRTA2005}.

\section*{Acknowledgments}
The first author is partially supported by the
ANR project 12IS02001 PACE. The second and third authors were supported by the project
ANR-2010-BLANC-021301 ``LOGOI''.
\bibliographystyle{abbrv}
\bibliography{biblio}

 \appendix
 
 \section{Appendix}
 
 \subsection{\SMELLY: Proof of Theorem~\ref{main_lem} }\label{app:SMELLY}

 In this section we prove Theorem~\ref{main_lem}. We focus on  \SMELLY\ nets, however it is important to observe that \emph{the 
 exact same constructions, results and proofs hold also for program nets (PCF nets in particular)} as defined in Section~\ref{sec:beyond}, because the number of structures associated to a $\bot$-node plays no role at any point; the proof that a reduction is always possible is only concerned with the nodes at the surface.
 
 All along this section, we assume that  $R$ is a simple \SMELLY\ net; therefore,  no symbol  $\bot$, $?$ or $!$ appears in the conclusions.
 We say that an axiom is \emph{polarized} if its conclusions are polarized formulas, \ie\ the axiom has the form $P\b,P \}$.  W.l.o.g, we assume that \emph{all axioms are atomic}; this hypothesis is not necessary, but it limits the number of cases in the proofs. All polarized axioms have therefore the form  $\{\bot,1\}$.
 We say that\emph{ a net $R$ is in SM normal form} if no sync reduction and no multiplicative reduction  is possible.\\
 %

  The proof of Thoeorem~\ref{main_lem} relies on the definition of a strict partial order on the set $\topnodes_R$, which we define as follows.
    \begin{center}
      $\topnodes_R :=$ the set of the following nodes at depth 0 in  $R$:
       \emph{sync nodes,  
         boxes, and polarized axioms}. 
     \end{center}
 We first observe the following two facts (both are immediate because of the typing of the nodes).

 \begin{lemma}[Sync Normal Forms]\label{sync_norm_form}
 If no $s$ reduction  applies, then the only nodes which can be  above a $sync$ node  are of type:  sync, \bbox, polarized axiom or  $\onelk$.
 \end{lemma}
 
 \begin{lemma}\label{bang_lem} Each edge of type $!A$ is conclusion of a box (exponential or \bbox).
 \end{lemma}

 The notion of \textit{non bouncing} path is well known in the literature of proof nets, and immediately extends to our case. Given a net $R$,
 a directed  path is \textit{non bouncing} if for each node on the path the following hold:
 \begin{itemize}

 \item\emph{ cuts and  axioms}: $r$ enters and exits from different  edges;
 
 \item  \emph{boxes}: if $r$ enters from an auxiliary conclusion, then exits from the principal conclusion, and viceversa;
 
 \item  $\otimes, \parr, sync $ \emph{nodes}: if $r$ enters from a premiss, then exits from a conclusion, and viceversa.
 
 \end{itemize}

  \begin{deff}[Priority path]

  We say that a non bouncing path $r$  is a \emph{ priority path} if $r$ starts from a node  $a\in \topnodes_R$ \emph{at depth 0} as follows:

 \begin{itemize}
 
 \item if $a$ is  a $sync$ node: $r$ exits $a$ from a premiss;
 
 \item if $a$ is a \bbox\ or a polarized axiom  $\{\bot, \one \}$:  $r$  {exits} $a$ from the principal conclusion $\bot$;
 
 \item if $a$ is an exponential box :  $r$  {exits} $a$ from an auxiliary conclusion.
 
 \end{itemize}

  \end{deff}
 
 We  observe that
 \begin{fact}
 A priority path has constant depth 0, as it  never enters boxes.
 \end{fact}
 
 We now study which nodes  a priority path can reach, and how.
 
 \begin{lemma}\label{priority_lem}
  Let  $R$ be a net  in SM normal form, and $r$ a priority path.
  
 \begin{itemize}
 \item When going \emph{downwards},  $r$ can reach only the following nodes, in the following way.
 
 \begin{itemize}
 
 \item   $\parr,\otimes, ?c, ?d$ nodes:  $r$  enters from a premiss, exits from  the conclusion.

 \end{itemize}
 
 \item When going \emph{upwards}, $r$ can reach only the following nodes,  in the following way.
 
 \begin{itemize}

 \item    $sync$ nodes:  $r$  enters from a conclusion, exits from  a premiss;
 \item  $\bot$ nodes:  $r$  enters  from an auxiliary conclusion (whose type is not $?A$), and exits from the principal conclusion;
 \item  $!$ and $Y$ nodes:  $r$  enters  from the principal conclusion, exits from an auxiliary conclusion;
 \item $\onelk$ nodes:  $r$  enters from the conclusion;
 \item  $r$  reaches  an \emph{axiom}  iff it is  of the form $\{\bot, \one\}$:  $r$  enters from $1$, exits from $\bot$;
 \item No $?w$ node can be reached.
 \end{itemize}
 
 \item Moreover,  going {downwards}, no edge of $r$  has type of the form $!A$, going  {upwards}, no edge of $r$  has type  $?A$.
 \end{itemize}

 \end{lemma}

 \begin{proof}

 We verify  the  lemma by induction on the length of $r$. Let us follow $r$, starting  from its origin  $a$, until either $r$ ends (in a conclusion, or in a node $\onelk$), or $r$ reaches a node   $l\in \topnodes_R$. We distinguish two cases.
 
 \begin{enumerate}
 \item  If $a$ is \emph{a box or an axiom, then  $r$  starts downwards} with  an edge of type $F$ where $F$ is either $\bot$ or  $?A$ (for some formula $A$). While descending, $r$ may traverse   $\otimes, \parr, ?c, ?d$ nodes (from a premiss to the conclusion);  $r$ cannot traverse any sync node, because $F$ is subformula of the type of each edge below $F$. We observe that no edge may have type $!B$, because of Lemma~\ref{bang_lem}.
  Descending, $r$ eventually reaches either a conclusion (in such a case the lemma is proved), or  a cut $c$ on which $r$ changes direction. Let $C$ be the premiss of $c$ which contains $F$, and $C\b$ its dual. It is immediate that  $C\b\not=?B$ (for any formula $B$), because otherwise we would have $C=!B\b$. Therefore, $C\b$ cannot be   conclusion of
   any $?d,?c, ?w$ node. 
     $C\b$ cannot be  conclusion of a node $\axlk$, $\otimes, \parr$,
  otherwise it is immediate to see that a multiplicative  reduction would apply.  
   As a consequence we have:
   \begin{enumerate}
 \item either $C\b$ is conclusion of a $\onelk$ node (hence the Lemma is verified);
 \item
     or $C\b$  is   conclusion of a sync node (which $r$ enters from a conclusion);
   \item
     or $C\b$ is  conclusion of  a  box $\BB$. If $\BB$ is a \bbox, $C\b$ must be an auxiliary conclusion (because $C\b$ contains $F\b$, \ie\  either $1$ or $!A\b$); if $\BB$ is an exponential box, $C\b$ must be the principal conclusion (because  $C\b\not=?B$ ).

   \end{enumerate}
 
 \item  If $a$ is a \emph{ sync node, then $r$ starts upwards}. By Lemma \ref{sync_norm_form}, any node $l$ above $a$ is either  a $\onelk$ node (hence the Lemma is verified), or   a node which belongs to $\topnodes_R$: either a sync node (which $r$ enters from a conclusion), or a  \bbox, which $r$ enters from an auxiliary  conclusion (because the edge  is positive), or a polarized axiom.\\
 
  Let us indicate by $r_1$ the  prefix of $r$ until the first node $l\in \topnodes_R$,  and by $r_2$ its continuation starting from $l$. We have verified that $r_1$ satisfies the  property. We observe also  that $r_2$ is again a priority path (of shorter length), and hence it satisfies the property by induction.
 
 \end{enumerate}
 \end{proof}
 The following  observation is immediate
 \begin{lemma}\label{switching_lem}
  Each priority    path is a  switching paths.
 \end{lemma}

 \begin{proof} A priority path is a path at constant  depth. For each $\parr$ or $?c$ node, $r$ only uses one premiss, because $r$ can only enter from a premiss  and exit from the conclusion. The dual is true for sync nodes.
 
 \end{proof}

 We are now able to set our main tool.
 
 \begin{prop}[Priority Order]\label{po} Let  $R$ be a net in SM normal form.
 The  relation $a\prec b$ for $a,b$ in  $\topnodes_R$ is defined if there is a priority  path $r$ from $a$ to $b$.
  This relation  defines a strict partial order on $\topnodes_R$, which we call \emph{priority order}.
 
 \end{prop}
 
 \begin{proof}
 We prove that the relation is
 
 \begin{enumerate}
 \item  Irreflexive: $a\prec a$ does \emph{not} hold for any $a$ in  $\topnodes_R$.
 
 \item  Transitive: $a\prec b$ and $b\prec c$ implies $a\prec c$.
 \end{enumerate}
 
 (1). By Lemma \ref{switching_lem}, $a\prec a$ would imply that there is a cyclic switching path.
 
 (2).  Let $r=k_1....k_n$  be the  priority  path from $a=k_1$ to $b=k_n$ and  $r'=n_1...n_m$ be the  priority path from $b=n_1$ to $c=n_m$. We claim that $b$
 is the only node that the two paths have in common, and we can hence concatenate  them and obtain a priority path from $a$ to $c$.
 Otherwise, assume that $l=n_j=k_i$ is the first node belonging to  $r'$ which belongs also to $r$.  We follow  $r'$ from $b$ to $l$ and $r$ from  $l$ to $b$.
 Let us call this path $p$, and check  that it is non bouncing on $l$. Therefore $p$ is  a priority path, in contradiction with the fact that  $b\prec b$ cannot hold.

 We observe  that  $p$ enters $l$ as $r'$ and  exits $l$ as $r$.
  $l$ cannot be a cut, otherwise $l$ would not be the first node which belongs to both paths.
 For all the other cases, Lemma \ref{priority_lem} guarantees that, if
 $l$ is a node of type $\otimes,\parr, ?c,?d$, then $r'$ enters from a premiss, and $r$ exits from a conclusion. The exact opposite is true if $l$ is a sync node. If $l$ is a \bbox\ $r'$ enters from an auxiliary conclusion,
 $r$ exits from the principal conclusion. The opposite is true in case $l$ is an exponential box.
  If $l$ is an axiom, it is polarized; $r'$ enters from a the positive conclusion, while
 $r$ exits from the negative conclusion.

 \end{proof}
 
 In order to prove Theorem~\ref{main_lem} we still need some technical lemmas.
 
 \begin{lemma}\label{maxbox}If $\BB$ is an exponential box, and $\BB$ is maximal for the priority order, then $\BB$ is a closed box.
 
 \end{lemma}
 \begin{proof} Each auxiliary conclusion $?A$ needs to be hereditary premiss of a cut. The path $r$ descending from $?A$ to the cut node $c$  is a priority path; the extension of $r$ with the other premiss $C$ of $c$ is still a priority path, which now is ascending . By Lemma \ref{priority_lem},  the source of  $C$  could be  either  a $\onelk$, which is not possible because of the type, or a node in $\topnodes_R$, against maximality of $\BB$. Therefore, $\BB$ cannot have any auxiliary conclusion.
 
 \end{proof}
 
 \begin{lemma}\label{cuts_lem} Let $R$ be a net in SM normal form.  $ \topnodes_{R}$ is  empty iff there  are no cuts.
 
 \end{lemma}
 
 \begin{proof} Assume $ \topnodes_{R}$ is  empty, then $R$ is an MLL net (with   $\onelk$ nodes also);   if there is a cut, we could perform a multiplicative reduction. Assume $ \topnodes_{R}$ is not  empty. If there is a box or a polarized axiom, its principal conclusion needs to be cut, because it does not appear in the conclusions.
 If there are sync nodes, but no boxes or axioms, we could apply an $s$ reduction. 
 
 \end{proof}

 \paragraph{Proof of Theorem~\ref{main_lem}}
 
 \begin{proof}
 Let $R$ be as in Theorem~\ref{main_lem}. 
  If $R$ is not in SM normal form, a  sync or  multiplicative reduction is possible by definition.
  If $R$ is in SM normal form, and contains  cuts, by  Lemma~\ref{cuts_lem}, $ \topnodes_{R}$ is non empty. We  find a valid reduction step by case analysis.
 
 \begin{itemize}
 \item  If  $ \topnodes_{R}$ contains  a \emph{maximal} node $l$ which is not an exponential box, we focus on it.
 
 \begin{itemize}
 \item
 \emph{$l$ is a sync node.} Any path moving upward from $l$ is a priority path.
 By using lemma \ref{sync_norm_form} and the fact that $l$ is maximal in $\topnodes_R$, we know that above $l$ there can only be
 $\onelk$ nodes. An $s.el$ reduction hence applies.
 \item
 \emph{$l$ is a \bbox\ or a polarized axiom.}  Let $\bot$ be the principal conclusion  of the box, or the negative conclusion of the axiom.  Since it cannot appear in the conclusions, $\bot$ must be hereditary premiss of a cut $c$.   We find $c$ by descending from $\bot$. Since  the path descending  from $l$  to the node $c$ is a priority path, by Lemma  \ref{priority_lem}
  we know that the first node entered by $r$ after the cut can only be conclusion of a $\onelk$, because any other possibility would belong to $\topnodes_{R}$ and is excluded
   by the maximality  of $l$. Hence a reduction applies (either $bot.el$ or $\axlk/\cutlk$)
 
 \end{itemize}
 
 \item  Otherwise, we choose a node $l$ as follows.
 \begin{itemize}
 \item  If  $ \topnodes_{R}$ contains only exponential boxes, we observe that all cuts have premisses of type $?A,!A$, and the $!A$ premiss is principal conclusion of an exponential box. Let $l$ be such a box.
 \item If $ \topnodes_{R}$ contains nodes which are not  exponential boxes, let $l$ be any such node.
 \end{itemize}

 If $l$ is already a  maximal exponential box, let $b_{max}:=l$, otherwise we  choose a  maximal exponential box $b_{max}$  such that $l \prec b_{max}$. 
 The key properties  that this careful  construction guarantees is that for each exponential box $\BB$ in the priority path from $l$ to $b_{max}$:
 \begin{itemize}
 \item[i.] the principal conclusion $!A$  of $\BB$ is premiss of a cut $c$;
 \item[ii.] the other premiss $?A\b$ of $c$ \emph{ is not auxiliary   conclusion} of a \bbox.
 \end{itemize}
  (i.) is true for $l$ by construction; moreover, for every exponential box $b$ which is reached by a priority path $r$, $r$ can enter   $b$ only from the  principal door, ascending from a cut  (see case (1.c) in the proof of Lemma~\ref{priority_lem}).  (ii.) is true for  any cut $c$   which is reached by  a priority path $r$, because if the cut has premisses $!A,?A\b $, $r$ can only  use the edge  $?A\b$ to 
  \emph{descend} in $c$, and must do so from a node which cannot be a   \bbox\ (because by  Lemma~\ref{priority_lem}, $r$ exits a \bbox\ only  from the $\bot$ conclusion). We are now able to conclude.
 
 By Lemma \ref{maxbox}, $b_{max}$ is a closed box.
 Because of (i.),  the  principal conclusion $!A$ of  $b_{max}$ is premiss of a cut $c$. 
 The other premiss of  $c$ has type $?A\b$, and because of (ii.), $?A\b$ can only be conclusion of a node of type $?d,?c,?w$, or auxiliary conclusion  of an exponential box. In each case  a closed reduction applies.
 \end{itemize}

 \end{proof}

 \subsection{\SIAM: Proofs }\label{app:SIAM}

 \begin{fact}[Parametricity]\label{fact:parametricity}
   Every transition rule is defined parametrically with respect to box stacks.
   That means, if a transition
 
 \noindent
   $\{(\edg_1, \stk_1, \bstk), (\edg_2, \stk_2, \bstk),
   \dots (\edg_n, \stk_n, \bstk)\} \cup T
   \redsiam
   \{(\edg'_1, \stk'_1, \bstk'), (\edg'_2, \stk'_2, \bstk'),
   \dots (\edg'_m, \stk'_m, \bstk')\} \cup T$
   is possible with box stack $\bstk$,
   then the transition 
 
 \noindent
   $\{(\edg_1, \stk_1, \overline{\bstk}), (\edg_2, \stk_2, \overline{\bstk}),
   \dots (\edg_n, \stk_n, \overline{\bstk})\} \cup T
   \redsiam
   \{(\edg'_1, \stk'_1, \overline{\bstk}'), (\edg'_2, \stk'_2, \overline{\bstk}'),
   \dots (\edg'_m, \stk'_m, \overline{\bstk}')\}
   \cup T$
   is also possible with box stack $\overline{\bstk}$.
 
 \end{fact}
 
 \begin{lemma}[Properties of $\trsf$]\label{lem:trsfProps}
   For any reduction $\netone \red \nettwo$,
   \begin{enumerate}
   \item If $\st \redsiam \sttwo$ in $\machine{\netone}$
     then $\trsf(\st) \redsiam^* \trsf(\sttwo)$ in $\machine{\nettwo}$.
     
     \item If $\st$ is an initial state in $\machine{\netone}$,
         then so is $\trsf(\st)$ in $\machine{\nettwo}$.
     
   \item If $\st \stopsiam$ in $\machine{\netone}$
     then $\trsf(\st) \stopsiam$ in $\machine{\nettwo}$.
   
   \item If $\st$ is a final state in $\machine{\netone}$,
     then so is $\trsf(\st)$ in $\machine{\nettwo}$.
   \item If $\st$ is a deadlock state in $\machine{\netone}$,
     then so is $\trsf(\st)$ in $\machine{\nettwo}$.
     
 %
   \end{enumerate}
 \end{lemma}
 \begin{proof}
   Each statement can be proved by case analysis.
   Note that statement $1.$ includes $\trsf(\st) = \trsf(\sttwo)$.
   \begin{enumerate}
   \item 
     If the transition $\st \redsiam \sttwo$ is not on the redex of 
     $\netone \red \nettwo$ then the claim holds,
     because the positions of tokens and the structure are the same except around the redex.
     Else we examine each case of reductions, where we have to consider only
     such a transition $\st \redsiam \sttwo$ that moves a token on the redex:
     \begin{itemize}
     \item \underline{$\red_{a}$}
       The states $\st$ and $\sttwo$ are mapped to
       $\trsf(\st) = \trsf(\sttwo)$ by definition of $\trsf$.
     \item \underline{$\red_{m}$} Similarly we verify that
       $\trsf(\st) = \trsf(\sttwo)$ (if the transition crosses $\tensor$ or $\parr$ node)
       or $\trsf(\st) \redsiam \trsf(\sttwo)$ (if the transition crosses $\cutlk$).
     \item \underline{$\red_{s}$} If the transition crosses the $\tensor$ node
       then $\trsf(\st) = \trsf(\sttwo)$.
       If the transition crosses the $\slk$ node then
       $\trsf(\st) \redsiam \st'_1 \redsiam \dots \redsiam \st'_n \redsiam \trsf(\sttwo)$,
       where the first transition crosses the $\slk$ node and each of the other transitions
       crosses the $\tensor$ node one by one.
     \item \underline{$\red_{s.el}$} Similar to the case of $\red_{a}$.
     \item \underline{$\red_{bot.el}$} If the transition crosses the $\cutlk$ node
       or enters the box then $\trsf(\st) = \trsf(\sttwo)$.
       Else, for the transition $\st \redsiam \sttwo$ in $\machine{\netone}$
       inside the box allowed by a stable token,
       there is always a transition $\trsf(\st) \redsiam \trsf(\sttwo)$
       since the structure contained in $\netone$ in the box is now surface in $\nettwo$.
     \item \underline{$\red_{c}$}
       If the transition crosses the $?c$ node then $\trsf(\st) = \trsf(\sttwo)$.
       Else $\trsf(\st) \redsiam \trsf(\sttwo)$ since the transition tules are defined
       parametrically with respect to box stacks
       (thus if $\st \redsiam \sttwo$ is done with a box stack $\bstk.l(\sigma)$,
       $\trsf(\st) \redsiam \trsf(\sttwo)$ can be done with box stack $\bstk.\sigma$.)
     \item For the other exponential rules the situation is similar:
       if the transition is on a dereliction token the states collapse,
       else $\trsf(\st) \redsiam \trsf(\sttwo)$ is possible by the same rule as
       $\st \redsiam \sttwo$ with a different box stack.
     \end{itemize}
 
   \item 
    Immediate.  Any token in an initial position is mapped to an initial position.

   \item 
   
      First we observe that  is impossible for tokens in $\st$ which are outside  the redex   to become able to move
          in $\trsf(\st)$, since their positions are  not modified.
          Thus we examine only the tokens in the redex.
            By case analysis, and by Fact~\ref{fact:parametricity}, the 
            existence of some state $\sttwo$ s.t.\ $\trsf(\st) \redsiam \sttwo$
            contradicts to $\st \stopsiam$.

   \item 
     Immediate. Any token in a  final position is mapped to a final position.
 
   \item 
   Immediate consequence of items 3. and 4.
     
   \end{enumerate}
 \end{proof}

 \begin{lemma}\label{lem:trsf_run}\quad\\
 \begin{itemize}
 \item[i.]  If  $\st_0 \to \cdots \to \st_n \redsiam \cdots$ is  a run of $\M_R$, then  $\trsf(\st_0) \to^{*} \cdots \to^{*} \trsf(\st_n) \redsiam \cdots$ is a run of $\M_{R'}$.   Moreover
 
 \item[ii.]   $\st_0 \to \cdots \to \st_n \redsiam \cdots$ is infinite/converges/deadlocks iff  $\trsf(\st_0) \to^{*} \cdots \to^{*} \trsf(\st_n) \redsiam \cdots$ does.
 
 \end{itemize}
 \end{lemma}
 \begin{proof}
 (i) is direct consequence of Lemma~\ref{lem:trsfProps}.   To prove (ii), we first  prove (1.) and (2.) below:
 \begin{enumerate}
 \item 
 \emph{If $\st_0 \to \cdots \to \st_n \redsiam \cdots $ is an infinite run then  $\trsf(\st_0) \to^{*} \cdots \to^{*} \trsf(\st_n) \redsiam \cdots$ is an infinite run}

   \begin{proof}
     \renewcommand{\qedsymbol}{$\blacksquare$}
 
   First, note that in every state $\st$ in a run, the set $\scod$ is finite.
   We call a transition $\st \redsiam \sttwo$ that satisfies $\trsf(\st) = \trsf(\sttwo)$
   (resp.\ $\trsf(\st) \redsiam^+ \trsf(\sttwo)$) a \emph{collapsing transition}
   (resp.\ a \emph{non-collapsing transition}).
   We show the following:
 \begin{itemize}
 \item 
 \emph{ For any state $\st$ s.t.\ $\st_0 \redsiam^* \st$,
   an infinite sequence of transitions $\st \redsiam \cdots$
   contains infinitely many non-collapsing transitions.
  }
 
   \noindent
   Let $\R \red_{a} \R'$, and $\edg_1,\edg_2,\edg_3$ be the edges in Figure~\ref{fig:trsf_re}.
   \begin{figure}[htbp]
     \centering
 
     \includegraphics{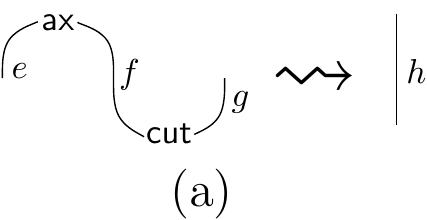}
     \caption{Edges of the redex}\label{fig:trsf_re}
   \end{figure}
   Since $\scod$ is finite, the set\linebreak
   \noindent
   $\{(\edg_i,s,\epsilon) \separ i \in \{1,2,3\} \} \cap \scod$ is finite.
   Let $n$ be the number of elements in this set of positions.
   It is straightforward to check that
   the length of a sequence of transitions from $\st$ only using collapsing transitions
   is bounded by $2n$. (We cannot apply more than two collapsing transitions on each token.)
   Thus an infinite sequence $\st \redsiam \cdots$ of transitions must contain
   a non-collapsing transition.
   By repeating this argument for all states in the run, we see that the infinite sequence $\st \redsiam \cdots$
   contains infinitely many non-collapsing transitions.
 \end{itemize}
 
   We conclude that there are infinitely many transitions
   $\st_i \redsiam \st_{i+1}$ s.t.\ $\trsf(\st_i) \redsiam^+ \trsf(\st_{i+1})$,
   and therefore the run $\trsf(\st_0) \to^{*} \cdots \to^{*} \trsf(\st_n) \to \cdots$
   is infinite.
   A similar argument applies for all other  reduction steps.
  
   \end{proof}
 
 \item\emph{ If $\st_0 \to \dots \to \st_n $ is a run which  terminates then  $\trsf(\st_0) \to^{*} \dots \to^{*} \trsf(\st_n)$ is a run which terminates.
 In this case, 
 \begin{enumerate}
 \item if $\st_n$ is final, so is $\trsf(\st_n) $
 \item if $\st_n$ is  not final, so is $\trsf(\st_n) $
 \end{enumerate}
 }
 \begin{proof}
  By using Lemma~\ref{lem:trsfProps}, items 3.,  4. and 5. 
 \end{proof}

 \end{enumerate}
 We conclude by noticing that (1.) and (2.) together  are  equivalent to either of the following:
 
 \begin{itemize}
   \item  $\st_0 \to \cdots \to \st_n \cdots $ is an infinite run iff   $\trsf(\st_0) \to^{*} \cdots \to^{*} \trsf(\st_n) \cdot$ is an infinite run
  
  \item  $\st_0 \to \dots \to \st_n $ is a run which  terminates iff   $\trsf(\st_0) \to^{*} \dots \to^{*} \trsf(\st_n)$ is a run which 
 
 \end{itemize}
 
 Similarly, (2.a) and (2.b) are equivalent to either of the following:
 
 \begin{itemize}
 \item  $\st_n$ is final iff  $\trsf(\st_n) $ is final;
 
 \item $\st_n$ is a deadlock iff  $\trsf(\st_n) $ is a deadlock.
 
 \end{itemize}
 \end{proof}
 
 \paragraph{Proof of Theorem~\ref{soundness}~(Soundess).}

 Let $R$ be a net of conclusions $A_1,\dots,A_n$, and
   $R\red R'$;  we adopt the
   following convention: we identify each conclusion of a net with the
   occurrence of formula $A_i$ typing it, so that there is no
   ambiguity. In particular, $R$ and $R'$ have\emph{ the same initial and
   final positions}.  We now show that the interpretation of a net is
   preserved by all normalization steps.
   
   We first observe the following.

 \begin{lemma}\label{orig}
   Let $\stI_R \redsiam \dots \redsiam \st $ a run in $\machine R$ and   $\stI_{R'} \redsiam \dots \redsiam \trsf (\st) $
     the corresponding  sequence of transitions in $\machine {R'}$. For each $\pp\in \scod$,
      we have the following:
     \begin{itemize}
     \item $\orig_R(\pp)\in \POSI_R$ iff $\orig_{R'}(\trsf (\pp))\in \POSI_{R'}= \POSI_R$
     \item if  $\orig_R(\pp)\in \POSI_R$ then  $\orig_{R}(\pp)= \orig_{R'}(\trsf (\pp))$  
     \end{itemize}
  
  \end{lemma}
 
 \begin{proof}
   By induction on the length of the run.
   If $\st = \stI_R$ it is immediate by definition of $\trsf$.
   If $\stI_R = \st_0 \redsiam \dots \redsiam \st_{n-1} \redsiam \st_n$,
   it is immediate to check that tracing back from the positions in $\trsf(\st_n)$
   reaches the same positions as $\trsf(\st_{n-1})$.
 \end{proof}

 We now can  prove that 
   $\sem{\netone} = \sem{\nettwo}$ as partial functions.

 \begin{proof}

 We know that $\POSI_R=\POSI_{R'}$ (let us call this set $\POSI$)   and  $\POSF_R = \POSF_{R'}$  (let us call this set $\POSF$). 
 Assume that $\M_R$ terminates in the state $\st$ and $\M_{R'}$ in the state $\st'$.
 We know  that $\trsf~ T = T'$ (because $\trsf~ T$ is a terminal final state of $\M_{R'}$, and such a state is unique);
 therefore  $ \scod \cap \POSF = \scod' \cap \POSF $. Let us call this set $X$. 
  It is immediate that if $\pp \in X$ then $\pp=\trsf ~\pp$. 
 Finally, by Lemma~\ref{orig}, for each $\ss\in \POSI$ and each $\pp \in X$,  $\orig_R(\pp) = \ss$ iff $ \orig_{R'}(\pp) =\ss$.
 From this we  conclude that $\sem{R}(\ss)= \sem{R'}(\ss)$.

 \end{proof}

\end{document}